\documentclass[journal]{IEEEtran}

\def\betaOL{\beta_{ \mbox{\scriptsize OL } } }
\def\betaC{\beta_{ \mbox{\scriptsize C } } }

\usepackage{graphics} 
\usepackage{epsfig} 
\usepackage{times} 
\usepackage{amsmath} 
\usepackage{amssymb}  
\usepackage{color}
\usepackage{setspace}

\newtheorem{theorem}{Theorem}
\newtheorem{lemma}{Lemma}
\newtheorem{proposition}{Proposition}

\newtheorem{remark}{Remark}

\definecolor{dblue}{rgb}{0,0,0}
\definecolor{mred}{rgb}{0,0,0}
\definecolor{dgreen}{rgb}{0,0,0}
\definecolor{mgreen}{rgb}{0,0,0}
\definecolor{mblue}{rgb}{0,0,0}
\definecolor{dyellow}{rgb}{0,0,0}
\definecolor{dmagenta}{rgb}{0,0,0}
\definecolor{dcyan}{rgb}{0,0,0}

\definecolor{dbblue}{rgb}{0, 0, 0}
\definecolor{drred}{rgb}{0, 0, 0}

\title{
Controlled Sensing for Multihypothesis Testing }

\author{

Sirin Nitinawarat, {\em Member, IEEE}, George Atia,
{\em Member, IEEE} and\\ Venugopal V. Veeravalli, {\em
Fellow, IEEE}

\thanks{This work was supported by the Air Force Office of Scientific Research (AFOSR) under the Grant
FA9550-10-1-0458 through the University of Illinois at
Urbana-Champaign, by the U.S. Defense Threat Reduction Agency through subcontract 147755 at the University of Illinois from prime award HDTRA1-10-1-0086, and by the National Science Foundation under Grant NSF CCF
11-11342.  The material in this paper was presented in
part at the IEEE International Conference on Acoustic,
Speech and Signal Processing, Kyoto, Japan, March 2012,
and at the IEEE International Symposium on
Information Theory, Cambridge, USA, July 2012.}

\thanks{S. Nitinawarat and V. V. Veeravalli are with the Department of Electrical and Computer Engineering and the Coordinated Science Laboratory, University of Illinois, Urbana, IL 61801.
Email: $\left \{ \right.$\texttt{nitinawa, vvv}$\left.
\right \}$\texttt{@illinois.edu.}}

\thanks{ 
G. K. Atia is with the Department of Electrical Engineering and Computer Science, University of Central Florida, Orlando, FL 32816-2362. Email: \texttt{george.atia@ucf.edu.} 
} }

\begin{document}

\maketitle

\begin{abstract}
The problem of multiple hypothesis testing with
observation control is considered in both fixed sample
size and sequential settings. 
In the fixed sample size setting, for binary hypothesis testing, the optimal exponent for the maximal error probability corresponds to the maximum Chernoff information over the choice of controls, and a pure stationary open-loop control policy is asymptotically optimal within the larger class of all causal control policies.
For multihypothesis
testing in the fixed sample size setting,  lower and
upper bounds on the optimal error exponent are derived.
It is also shown through an example with three
hypotheses that the optimal causal control policy can
be strictly better than the optimal open-loop control
policy.  In the sequential setting,  a test based on
earlier work by Chernoff for binary hypothesis testing,
is shown to be first-order asymptotically optimal for
multihypothesis testing in a strong sense, using the 
notion of decision making risk in place of the overall 
probability of error.  Another test is also designed to 
meet hard risk constrains while retaining asymptotic 
optimality.  The role of past information and
randomization in designing optimal control policies is
discussed.

\end{abstract}

\section{Introduction}
\label{sec:intro} The topic of controlled sensing for
inference deals primarily with adaptively managing and
controlling multiple degrees of freedom in an
information-gathering system, ranging from the sensing
modality to the physical control of sensors, to achieve
a given inference task. 
Unlike in traditional
control systems, where the control primarily affects
the evolution of the state, in controlled sensing, the
control affects  only the observations. 
The goal is for the decision-maker to infer the state
accurately by shaping the quality of observations.

Some applications of controlled sensing
include, but are by no means limited to, target
detection, tracking and classification (see, e.g.,
\cite{cast-cdc-1997, atia-veer-fuem-ieeetsp-2011}).  
Controlled sensing policies were also developed for
landmine and underwater mine classification in
\cite{hero-et-al-book-2008}.  In the domain of sequential clinical
trials, controlled sensing has been used to planning 
of medical trials under an ethical injunction against 
unnecessary continuance of inferior treatments
\cite{ansc-jasta-1963}. Dynamic sensor selection and
scheduling policies were also developed for tracking
and target localization in
\cite{he-chon-dsp-2006, kris-djon-ieeetsp-2007}.

In this paper, we focus on the basic inference problem
of hypothesis testing, and our goal is to find an asymptotically 
optimal joint-design of a control policy and a decision rule
(in addition to a stopping rule for the sequential
setting) to decide among the various hypotheses
\cite{niti-atia-veer-icassp-2012}.  In particular, we consider 
a Markovian model for the
simple hypothesis testing of multiple hypotheses with
observation control.  Prior to making a decision about
the hypothesis, the decision-maker can choose among
different actions to shape the quality of the
observations.  
We consider both the fixed
sample size and sequential settings of this problem. In
the latter setting, the controller can adaptively
choose to stop taking observations, and the sequential
test is fully described by a control policy, a stopping
rule and a final decision rule.

\subsection{Relationship to  Prior Work}

We begin by discussing prior work in the fixed
sample size setting. Tsitsiklis  \cite{tsit-mcss-1988}
considered the problem of quantizing
independent observations at geographically
separated sensors for multiple hypothesis testing. The
number of sensors, which is taken to infinity in
\cite{tsit-mcss-1988}, can be considered to be
equivalent to the sample size in our controlled sensing
problem. Therefore, the quantization rules can be
considered to be special  cases of control actions that
can affect the observations at the output of the
various sensors. However, the control actions in
the controlled sensing problem are much
more general.
Furthermore, the observation control policy in
\cite{tsit-mcss-1988} is effectively an open-loop
control policy.  In contrast, our main focus in this
paper  is on temporal observation control in which the
control at each time can be influenced by the past
observations.

In the fixed sample size setting, the block
channel coding problem with feedback and with a
{\em fixed number of messages} studied by Berlekamp
\cite{berl-phd-thes-1964} can also be considered to be
a special case of the controlled sensing
problem. This is because, in the coding problem,
the controller (encoder) has access to the hypothesis
(message), whereas in our controlled sensing problem
the controller is not assumed to have access to the
hypothesis and is therefore more challenging. 

The controlled sensing problem is also
more general than the
multi-channel identification problem treated
by Mitran and Kav\u{c}i\'{c}
\cite{mitr-kavc-ieeetit-2006}, in which there is a {\em
finite} constraint on the number of past channel
outputs
available to the input signal selector
at each time.
In contrast, the causal control policies
considered herein
can depend on the entire past observations, the
number of which becomes {\em
unbounded} as the horizon approaches infinity.  In
related work, Hayashi \cite{haya-ieeetit-2009}
considered the adaptive discrimination of channels with unbounded memory,
but for  only {\em two} channels,
i.e., two hypotheses.

In Section \ref{section-FSS}, we first present a characterization for the optimal error
exponent for binary hypothesis testing with a fixed
sample size showing that a pure stationary open-loop
control, where the control value at each time is fixed
and does not depend on past measurements and past
controls, achieves the optimal error exponent among the
class of causal controls.
In fact, this result is in agreement with
that of Hayashi on discrimination of two channels
\cite{haya-ieeetit-2009}
(see also Footnote \ref{footnote-1}).
Then, for general multiple hypothesis testing with a
fixed sample size, we derive a characterization for the
optimal error exponent achievable by open-loop control.
With more than two hypotheses, the characterization for
the optimal error exponent achievable by causal control
(which can be a function of past measurements and past
controls) is a more difficult problem.  
Nevertheless, we show through a concrete
example with {\em only} three hypotheses that the
optimal causal control policy can be strictly better
than the optimal open-loop control policy.  We also
derive general lower and upper bounds for the optimal
error exponent achievable by causal control.

We now discuss related work in the sequential
setting. The problem of  sequential hypothesis testing
\emph{without control} was introduced by Wald
\cite{wald-seq-book-1947, wald-wolf-amstat-1948} and
studied in detail for the binary hypothesis case.  In
this work, the optimal expected values of the stopping
time
were characterized subject to constraints on the
probabilities of error under each hypothesis.  It was
shown that the Sequential Probability Ratio Test (SPRT)
is optimal, i.e., among all tests with the same power,
the SPRT requires on average the fewest
number of observations.  An extension to the
multihypothesis case was
considered in \cite{baum-veer-ieeetit-1994} where the
authors proposed a Multihypothesis SPRT (or MSPRT)
which was later shown to satisfy certain asymptotic
optimality conditions \cite{ veer-baum-ieeetit-1995,
drag-et-al-ieeetit-1999, drag-et-al-ieeetit-2000}.

The problem of sequential binary composite hypothesis
testing \emph{with observation control} was considered
by Chernoff \cite{cher-amstat-1959} and an
asymptotically optimal sequential test was presented.
While Wald's SPRT is optimal in the sense that it
minimizes the 
expected values of the stopping time
among all tests 
for which the probabilities of error do
not exceed predefined thresholds
\cite{wald-wolf-amstat-1948}, a weaker notion of
optimality is 
adopted
in \cite{cher-amstat-1959}.
Specifically, the proposed test is shown to achieve
optimal expected values of the stopping time subject to
the constraints of {\em vanishing} probabilities of
error under each hypothesis. 
The sequential
test with causal control proposed by Chernoff can only
be proven to be asymptotically optimal
under under a set of positivity constraints
on the Kullback-Leibler distances as defined in
(\ref{eqn-SS-positivity-cond}).
Bessler \cite{bess-tech-repo-1960} generalized
Chernoff's work to general multiple hypothesis testing
but also imposed the same type of assumption on the
model.\footnote{We would like to
thank an anonymous reviewer for pointing us to the
generalization of Chernoff's test to the $M>2$ case in
Bessler's dissertation.}

Burna\v{s}ev \cite{burn-math-ussr-izve-1980} considered
the problem of sequential discrimination of multiple
hypotheses with control of observations under a
different information structure. It is important to
note that the controlled sensing problem that we
consider is fundamentally different from Burna\v{s}ev's
problem. Unlike \cite{burn-math-ussr-izve-1980}, where
the control actions are functions of the underlying
hypothesis, in \cite{cher-amstat-1959} and the setting
we consider herein the control actions cannot be
functions of the unknown hypothesis. In that sense, the
problem considered in \cite{burn-math-ussr-izve-1980}
has a simpler structure since the controller knows the
underlying hypothesis. This knowledge simplifies both
the optimization of control policies as well as their
performance analysis. When the hypothesis is unknown to
the controller, as in the controlled sensing problem
considered herein, the controller has to base its
control actions on estimates of the unknown hypothesis.

A Bayesian version of this sequential problem (with
observation control) was considered by the authors in
\cite{nagh-javi-isit-2010} in the non-asymptotic
regime. Since the optimal policy is generally difficult
to characterize, certain conditions (Blackwell ordering
\cite{blac-amstat-1953}) were identified under which
the optimal control is an open-loop control. The main
focus of \cite{nagh-javi-isit-2010, nagh-javi-asilomar-2010, 
nagh-javi-isit-2011, nagh-javi-ciss-2012} has been on trying 
to solve the underlying dynamic program and finding the
structure of optimal solutions, a task that is
only possible
in some special cases. In contrast, our work mostly
focuses on performance analysis and on establishing the
asymptotic optimality of proposed control policies.

In Section \ref{section-SS}, we extend the
results in \cite{cher-amstat-1959, bess-tech-repo-1960}
in several directions. First, we show that the
sequential test in \cite{cher-amstat-1959,
bess-tech-repo-1960} is asymptotically optimal
in a {\em strong} sense, using the notion of frequentist risks
in place of the probability of error. Second, we
dispense with the positivity assumption on the Kullback-Leibler divergences used in \cite{cher-amstat-1959,
bess-tech-repo-1960},   by constructing a modification to Chernoff's test that does not require this assumption. Third, we construct a further modification to the test that meets hard constraints on the frequentist risks, while
retaining asymptotic optimality.

\subsection{Paper Outline}

The remainder of the paper is organized as follows. In
Section \ref{section-Prelim}, we specify the general
notations and assumptions that will be adopted
throughout the paper. Our problem formulations and
results for the fixed sample size setting and the
sequential setting, together with a summary
of our contributions in each case, are given in
Section \ref{section-FSS} and \ref{section-SS},
respectively;
An example is provided in Section
\ref{section-Example}. A discussion is  provided in
Section \ref{section-Discussion}, and conclusions are
given in Section \ref{section-Conclusion}. 

\section{ Preliminaries }
\label{section-Prelim}

Throughout the paper, random variables are denoted by
capital letters and their realizations are denoted by
the corresponding lower-case letters.

Consider hypothesis testing with $M$ hypotheses, with
the set of hypotheses denoted by $\,\mathcal{M}
\triangleq \left \{ 0, \ldots, M-1 \right \}.$  At each
time step, the observation takes values in
$\,\mathcal{Y}\,$ and the control takes values in
$\,\mathcal{U}\,$. We assume that the control alphabet
$\,\mathcal{U}\,$ is {\em finite}. The observation
alphabet $\,\mathcal{Y}\,$ is a measurable space; it
can be either continuous, i.e., a finite-dimensional
Euclidean space, or discrete. Under each hypothesis
$\,i \in \mathcal{M},\,$ and at each time $k,$
conditioning on the event that the current control
$u_k$ has value $u,$ the current observation $Y_k$ is
assumed to be conditionally independent of past
observations and past controls $\left( y^{k-1}, u^{k-1}
\right) \triangleq \left( \left( y_1, \ldots, y_{k-1}
\right), \left( u_1, \ldots, u_{k-1} \right) \right).$

\noindent
We refer to this (conditionally) memoryless assumption
as the stationary Markovity assumption.

The following technical assumptions are made throughout
the paper. First, for every $\,u \in \mathcal{U},\,$ we
assume that the distributions of the observations under
each hypothesis $i \in \mathcal{M}$ are absolutely
continuous with respect to a common distribution
$\mu_u$ on $\,\mathcal{Y}.\,$  Consequently, for every
$u \in \mathcal{U}$ and every $\,i \in \mathcal{M},\,$
there exists a probability density function
(pdf)/probability mass function (pmf) $p_i^u$
such that for
every measurable set $\,A \subseteq \mathcal{Y},$
\begin{align}
    \mathbb{P}_i^u \left \{ Y \in A \right \} \ =\  \int_{ A } ~p_i^u \left( y \right) \,d \mu_u \left( y \right),\ \
    u \in \mathcal{U},
\label{eqn-assumption-1}
\end{align}
where the notation $\mathbb{P}_i^u$ denotes the
probability measure with respect to the distribution
$p_i^u$. Second, we also assume that for every $u \in
\mathcal{U}$ and every pair $i, j \in \mathcal{M},\ i
\neq j,$
\begin{equation}
    \mathbb{E}_i^u \left[ \left( \log{ \left( \frac{p_i^u \left( Y \right)}{p_j^u \left( Y \right)} \right)} \right)^2 \right ]
    \ <\ \infty,
\label{eqn-assumption-2}
\end{equation}
where the notation $\mathbb{E}_i^u$ denotes an
expectation with respect to $p_i^u.$  Note that it
follows from (\ref{eqn-assumption-2}) that for every $u
\in \mathcal{U}$ and every pair $i, j \in \mathcal{M},\
i \neq j,\ p_i^u$ is absolutely continuous with respect
to $p_j^u$.  However, for $u, u' \in \mathcal{U},\ u
\neq u',$ and $i, j \in \mathcal{M},\ p_i^u$ need not
be absolutely continuous with respect to $p_j^{u'}.$
For a finite $\mathcal{Y},$ the combination of
(\ref{eqn-assumption-2}) and the first assumption is
tantamount to the condition that all pmfs in the
collection $\left \{ p_i^u \right \}_{i \in
\mathcal{M}}$ have the same support. However, the
support could be different for different values of $u$.

\section{Fixed Sample Size Setting}
\label{section-FSS}

In this section, we first consider the setting wherein
the sample size is fixed a priori, i.e., it does not
depend on specific realizations of the observations and
controls.

We consider two classes of control policies based on
two information patterns.  The first is the open-loop
control policy where the (possibly randomized) control
sequence $\left( U_1, \ldots, U_n \right)$ is assumed
to be independent of the observations $\left( Y_1,
\ldots, Y_n \right).$ The second is the causal control
policy where at each time $k,$ the control $U_k$ can be
any (possibly randomized) function of past observations
and past controls, i.e., $U_k,\ k = 2, 3, \ldots, n,$
is described by an arbitrary conditional pmf $q_k
\left( u_k \vert y^{k-1}, u^{k-1} \right),$ and $U_1$
is distributed according to a pmf $q_1 \left( u_1
\right)$.  If all these (conditional) pmfs are
point-mass distributions, i.e., the current control is
a deterministic function of past observations and past
controls, then the resulting policy is a {\em pure}
control policy.  Under the aforementioned stationary
Markovity assumption,
the joint probability distribution function of 
$\left(
Y^n, U^n \right)$ under each hypothesis $i,$ denoted by
$p_i \left( y^n, u^n \right),$ can be written as
\begin{align}
    p_i \left( y^n, u^n \right)
    \,\triangleq\, q_1 \left( u_1 \right)
    \prod_{k=1}^n p_i^{u_k} \left( y_k \right)
    \prod_{k=2}^n q_k \left( u_k \vert y^{k-1}, u^{k-1} \right).
    \label{eqn:FSS-CSdistribution-hypo-i}
\end{align}
For open-loop control, $q_k \left( u_k \vert y^{k-1},
u^{k-1} \right)$ is (conditionally) 
independent
of $y^{k-1}$;
hence,
\begin{align}
p_i \left( y^n, u^n \right)
    &~=~
    \left(
    \prod_{k=1}^n p_i^{u_k} \left( y_k \right)
    \right)
    \left(
    q_1 \left( u_1 \right)
    \prod_{k=2}^n q_k \left( u_k \vert u^{k-1} \right)
    \right)	\nonumber \\
    &~=~ \left(
    \prod_{k=1}^n p_i^{u_k} \left( y_k \right)
    \right)
    q \left( u^n \right).
    \label{eqn:FSS-OLdistribution-hypo-i}
\end{align}

After $n$ observations, a decision is made about the
hypothesis according to the rule $\,\delta:
\mathcal{Y}^n \times \mathcal{U}^n \rightarrow
\mathcal{M}\,$ with maximal error probability: $e \left(
\left \{q_k \right \}_{k=1}^n, \left \{ p_i^u \right
\}_{i \in \mathcal{M}}^{u \in \mathcal{U}}, \delta
\right) \triangleq \max\limits_{i \in \mathcal{M}}
~\mathbb{P}_i \left \{ \delta \neq i \right \}. $ Note
that for a pure control policy, $u^n$ is either a fixed
sequence (pure open-loop control) or a deterministic
function of the observations $y^n$ (pure causal
control).  Consequently, when a pure control policy is
adopted, it suffices to consider a decision rule that
is a function only of the observations, i.e., $\delta
\left( y^n, u^n \right) = \delta \left( y^n \right).$
The combination of a control policy and a decision rule
will be referred to as a {\em test}. The asymptotic
quantities of interest will be the largest exponent of
the maximal error probability achievable by open-loop
control, denoted by $\betaOL,$ and by causal control,
denoted by $\betaC$, respectively.  In particular,
\begin{align}
    \hspace{-0.15in}
    \betaOL  \triangleq~
    \mathop{\overline{\mbox{lim}}}_{n}\
    \sup_{\delta,\ q \left ( u^n \right)}\
    \ - \frac{1}{n} \log
        \left(
            e \left( q \left( u^n \right), \left \{ p_i^u \right \}_{i \in \mathcal{M}}^{u \in \mathcal{U}}, \delta \right)
        \right);
    \nonumber
\end{align}
\begin{align}
\betaC 
~\triangleq~
\mathop{\overline{\mbox{lim}}}_{n}\
    \hspace{-0.30in}
    \mathop{\sup_{
    \delta,\ q_1 \left( u_1 \right)}}_
    {\left \{ q_k \left( u_k \vert y^{k-1}, u^{k-1} \right) \right \}_{k=2}^n }
    \hspace{-0.35in}
    - \frac{1}{n} \log
        \left(
            e \left( \left \{ q_k \right \}_{k=1}^n, \left \{ p_i^u \right \}_{i \in \mathcal{M}}^{u \in \mathcal{U}}, \delta \right)
        \right). \nonumber
\end{align}
It follows immediately from these definitions that
$\betaOL \leq \betaC,$ as the information pattern
associated with causal control is more informative than
that associated with open-loop control.  We also seek
to characterize the optimal control policies that
achieve the optimal error exponents.

Note that because the number of hypotheses is fixed, we
can consider a Bayesian probability of error (with
respect to any prior probability distribution of the
hypothesis) instead of the maximal one in the
definitions of the optimal error exponents without
changing their optimal values.

Before moving on to the technical part, we
first summarize our contributions in this section.
\begin{itemize}
\item  We derive a characterization for the optimal error exponent achievable by open-loop control for general multiple hypothesis testing with a fixed sample size (see Remark \ref{remark-openloop-FSS} explaining the connection between this result and  previous work \cite{mitr-kavc-ieeetit-2006}).
\item We propose a test for general multiple hypothesis testing with a fixed sample size using a causal control policy that chooses the control value based on a 
suitable
Chernoff information.  We also derive general lower and upper bounds for the optimal error exponent achievable by causal control that holds for any number of hypotheses, and illustrate through a canonical example with {\em only three} hypotheses that causal control can outperform open-loop control.
\end{itemize}

\subsection{The Case of Binary Hypothesis Testing $\left(
M = 2 \right)$}

For $p_1$ and $p_2$ that are pdfs/pmfs on $\mathcal{Y}$ with respect to a common distribution $\lambda$, the
Kullback-Leibler (KL) distance
of $p_1$ and
$p_2,$ denoted by $D\left( p_1 \| p_2 \right),$ is
defined as
\begin{equation}
    D \left( p_1 \| p_2 \right) \ \triangleq \
        \int_{y} ~p_1(y) \log{\left( \frac{p_1(y)}{p_2(y)} \right)} \,d \lambda \left( y \right).
    \nonumber
\end{equation}

We start with the following
characterizations for the largest error exponents
achievable by open-loop control and by causal control
in the case of binary hypothesis testing.

For any $u \in \mathcal{U}$ and any $s \in [0, 1]$,
consider the following pdf/pmf
\begin{align}
    b_s^u \left( y \right) &\ \triangleq\
    \frac{p_0^u \left( y \right)^s p_1^u \left( y \right)^{1-s}}
    {\int_{\overline{y}} ~p_0^u \left( \overline{y} \right)^s p_1^u \left( \overline{y} \right)^{1-s}   \,d \mu_u \left( \overline{y} \right)
    },
    \label{eqn-FSS-projected-distribution} 
\end{align}
and also let
\begin{align}
s^* \left( u \right) &\ \triangleq\ 
\mathop{\mathrm{argmax}}_{s \in [0, 1]} \ -\log \left(
\int_y ~ p_0^u \left( y \right)^s
        p_1^u \left( y\right)^{1-s} \,d \mu_u \left( y \right)
  \right).
  \nonumber
\end{align}

\begin{proposition}
\label{prop-1} For $M = 2,$ it holds that\footnote{
\label{footnote-1} Although this result is mathematically equivalent to \cite[Theorem
1]{haya-ieeetit-2009} on discrimination of two channels, we point out here that
 the ``discrimination'' problem is motivated by the channel coding 
problem (see  \cite{burn-math-ussr-izve-1980}), in which the controller
can be considered to know the true hypothesis. }

\begin{align}
    \betaOL &~=~ \betaC \nonumber \\
    &~=~ \max_{u \in \mathcal{U}}\
    \max_{s \in \left [ 0, 1 \right ]}
\,-\log \left( \int_y p_0^u \left( y \right)^s p_1^u
\left( y \right)^{1-s} d \mu_u \left( y \right) \right)
    \label{eqn:FSS-opt-Chernoff-exponent-A} \\
    &~=~ \max_{u \in \mathcal{U}} \
            D \left( b_{s^* \left( u \right)}^u \| p_0^u  \right)
    ~=~ \max_{u \in \mathcal{U}} D \left( b_{s^* \left( u \right)}^u \| p_1^u \right).
    \label{eqn:FSS-opt-Chernoff-exponent-B}
\end{align}
\end{proposition}

\begin{remark}
For each fixed $u \in \mathcal{U}$, the quantity
\begin{equation}
C \left( p_0^u, p_1^u \right) ~\triangleq~
\mathop{\mathrm{max}}_{s \in \left [ 0, 1 \right ]}
~-\log \left( \int_y ~p_0^u \left( y \right)^s p_1^u
\left( y \right)^{1-s} \,d \mu_u \left( y \right)
\right) \label{eqn-def-Chernoff-information}
\end{equation}
is called the ``Chernoff information'' of $\,p_0^u\,$
and $\,p_1^u.\,$
Consequently, 
Proposition \ref{prop-1}
(cf. (\ref{eqn:FSS-opt-Chernoff-exponent-A})) states that
the optimal error exponent
is the maximum Chernoff information over the choice of
controls.
\end{remark}
\begin{remark}
It follows from Proposition \ref{prop-1} and the
result on the Chernoff information for i.i.d.
observations that the above optimal error exponent is
achievable by a pure stationary open-loop control
sequence in which, for every $k = 1, \ldots, n,\ u_k =
u^*,$ where $u^*$ is the maximizer associated with the
right-side of (\ref{eqn:FSS-opt-Chernoff-exponent-A})
(or, identically, with the two (maximizing)
optimization problems in
(\ref{eqn:FSS-opt-Chernoff-exponent-B})).  In
particular, {\em information from the past and
randomization are superfluous for attaining the best
error exponent for fixed sample size binary hypothesis testing.}
\end{remark}

\subsection{The Case of Multiple Hypothesis Testing $\left(
M > 2 \right)$}

\subsubsection{Open-loop Control}

Our first theorem pertains to the situation
with open-loop control.

\vspace{0.05in}
\begin{theorem}
\label{Thm2} For $M > 2,$ it holds that
\begin{align}
 \betaOL ~=~ \hspace{3.02in}		\nonumber \\
    \max_{q \left( u \right)}
    \min_{i \neq j}  	
    \max_{s \in \left [ 0, 1 \right ]}
    \hspace{-0.03in}
    -
    \hspace{-0.03in}
    \sum_{u}  q \left( u \right)
           \log \left(
           \int_{y} p_i^u \left( y \right)^s
           p_j^u \left( y \right)^{1-s} d \mu_u \left( y \right)
           \right), \label{eqn:FSS-multiplehypo-opt-OLexponents}
\end{align}
where the left-most maximization is over all pmfs $q$
on $\mathcal{U}$ and the minimization is over all pairs
of hypotheses $i, j,\, i \neq j$. Furthermore,
$\betaOL$ is
achievable by pure (non-randomized) control.
\end{theorem}
\begin{remark}
\label{remark-openloop-FSS}
For finite observation alphabets, $ \betaOL$ 
in (\ref{eqn:FSS-multiplehypo-opt-OLexponents}) can be shown to be equal to the alternative formula derived in \cite[Theorem 5]{mitr-kavc-ieeetit-2006}.  However, our formula in (\ref{eqn:FSS-multiplehypo-opt-OLexponents}) is simpler than that in \cite[Theorem 5]{mitr-kavc-ieeetit-2006} because it involves maximization over a single real-valued spurious parameter $s$ instead of minimization over a conditional distribution as in \cite[Theorem 5]{mitr-kavc-ieeetit-2006}.  More importantly, our result applies also to general observation alphabets not just the finite ones.
\end{remark}

\subsubsection{Causal Control}
\label{subsubsec-FSS-CS}

A natural question that arises now
is whether
causal control can yield a larger error exponent than
open-loop control when $M > 2$.  
The answer will be shown to be affirmative
even for $M = 3$.  To this end,
we now propose a test with 
{\em pure} causal control
(we show in Theorem \ref{Thm3} below
that pure causal control
does
achieve the optimal error exponent).

Our test admits the following recursive description and
is based on the use of the posterior distribution of
the hypothesis as a sufficient statistic.
Having obtained the first 
$k$ observations
$y^k,$ we find
the maximum likelihood (ML) estimate of
the hypothesis, denoted by $\hat{i}_k \left( y^k
\right) = \mathop{\mbox{argmax}}_{i \in \mathcal{M}}~
p_i \left ( y^k \right ).$\footnote{In case of ties, we
pick, say, the hypothesis with the least numerical
value.} We adopt a pure control policy wherein $u_{k+1}
\in \mathcal{U}$ is selected as
\begin{align}
    u_{k+1} \,=\, u_{k+1} \left( \hat{i}_k \right)
    \,=\, \mathop{\mbox{argmax}}_{u \in \mathcal{U}}
    \min_{j \in \mathcal{M} \backslash \left \{ \hat{i}_k \right \} }
    C \left( p^u_{\hat{i}_k}  , p^u_j  \right),
    \label{eqn-FSS-opt-control}
\end{align}
where $C \left( p^u_{\hat{i}_k}  , p^u_j  \right)$ is
the Chernoff information of $\,p^u_{\hat{i}_k}\,$ and
$\,p^u_j\,$ defined in
(\ref{eqn-def-Chernoff-information}). Lastly, at the
final time $n,$ the decision rule is specified as
$\,\delta \left( y^n, u^n \right) ~=~ \hat{i}_n.\,$
{\em The proposed test follows the celebrated
separation principle between estimation and control;
while estimating the ML hypothesis is carried out
online, the control is chosen based on a stationary
deterministic mapping from the space of posterior
distributions to the control space, and hence, the
mapping can be fully specified offline.} It will be
shown in Section \ref{section-Example} that for the
special example with {\em only} three hypotheses, this
proposed test is superior to the best open-loop
control.  In general, we still do not know the
structure of the optimal causal control, and
characterizing the optimal error exponent for causal
control is a hard problem even for $M = 3$.
Nevertheless, we derive precise bounds on the optimal
error exponent that are applicable for any $M > 2$.
Note that the optimal error exponent achievable by
open-loop control as characterized in Theorem
\ref{Thm2} already serves as a lower bound for the
optimal error exponent achievable by causal control. We
also derive a new lower bound and an upper bound for
the optimal error exponent for causal control. These
bounds are stated in Theorem \ref{Thm3} for the fixed
sample size setting with $M > 2$. Although the lower
bound of Theorem \ref{Thm3} for $\betaC$ holds only for
a {\em finite} observation alphabet $\mathcal{Y}$, the
upper bound in Theorem \ref{Thm3} and all the previous
results are valid for an arbitrary $\mathcal{Y}$
(subject to assumptions (\ref{eqn-assumption-1}) and
(\ref{eqn-assumption-2}) in Section
\ref{section-Prelim}). As mentioned in Section
\ref{section-Prelim}, for a finite $\mathcal{Y}$, we
assume that for every $u \in \mathcal{U},$ the
collection of pmfs $\left \{ p_i^u \right \}_{i \in
\mathcal{M}}$ have the same support.

\vspace{0.05in} For any pmf $\nu$ on $\mathcal{M},$ any
$u \in \mathcal{U},$ let $\nu \circ p^u \left( \cdot
\right)$ denote the pmf/pdf (on $\mathcal{Y}$) $~\sum_i
\nu \left( i \right) p_i^u \left( y \right)$.
\begin{theorem}
\label{Thm3} For every finite $\mathcal{Y}$ and every
$M > 2,$ it holds that
\begin{align}
    \sup_{\eta > 0} \
    - \log{
    \left (
    \sup_{\nu}
    \min_{u}
    \max_i
    \left(
    \sum_{y} p_i^u \left( y \right)
    e^{
    \frac{
    \eta \left [
                \nu \circ p^u \left( y \right) - p_i^u \left (y \right)
                \right ]
    }
    {\left( 1 - \nu \left( i \right) \right)}
    }
    \right)
    \right )
    } \hspace{0.11in}
    \nonumber \\
    \hspace{-0.2in}
    \leq\ \ \ \betaC  \ \ \ \leq  \hspace{2.5in} 
    \nonumber 
\end{align}

\vspace{-0.3in}
\begin{align}
    \hspace{0.2in} 
    \min_{i \neq j}
    ~\max_{u}
    ~\max_{s \in \left [ 0, 1 \right ]}
           -\log \left(
           \sum_{y} p_i^u \left( y \right)^s
           p_j^u \left( y \right)^{1-s}
           \right), 
    \label{eqn:FSS-multiplehypo-opt-exponents}
\end{align}

\noindent where the outer supremum for the argument of
$-\log{}$ in the lower bound is over pmfs $\nu$ on
$\mathcal{M}$ that are not point-mass distributions and
the outer minimization for the upper bound in
(\ref{eqn:FSS-multiplehypo-opt-exponents}) is over all
pairs of hypotheses $i, j,\ i \neq j$. Furthermore, as
for $\betaOL$, the exponent $\betaC$ is also achievable
by pure control without any randomization.
\end{theorem}

\begin{remark}
The optimization problem for the lower bound in (\ref{eqn:FSS-multiplehypo-opt-exponents}) can be handled off-line.  In the example below, we show that the value of this lower bound is strictly larger than $\betaOL$.
\end{remark}

\section{Sequential Setting}
\label{section-SS} 

In the previous section, we considered tests with a
fixed sample size.  In this section, we consider a
different setting in which the controller can
adaptively decide, based on the realizations of past
observations and past controls, whether to continue
collecting new observations, thereby deferring making a
final decision about the hypothesis until later time,
or to stop taking observations and make the final
decision. In this setting, the goal is to design a {\em
sequential} test to achieve the optimal tradeoff
between reliability, in terms of probability of error,
and delay or cost, in terms of the expected sample size
needed for decision making.  Unlike in the fixed
sample size setting in which the asymptotic analysis of
tests with open-loop control is easier than that of
tests with causal control, in the sequential setting,
the contrary situation seems to hold.  
In particular,
as we show below, the adoption of randomized {\em causal}
control in the sequential setting enables the
simultaneous minimization of the expected sample sizes
under the $M$ hypotheses as the error probability
vanishes.
In contrast, {\em open-loop} control does not enable 
such simultaneous minimization, and therefore the characterization of the 
associated  tradeoff is more difficult and remains open.
This is why
we only consider causal control in the sequential
setting.

We now summarize our contributions in this setting.
\begin{itemize}
\item  The existing sequential test originally proposed by
Chernoff \cite{cher-amstat-1959} for binary composite
hypothesis testing, and extended to the multihypothesis
setting by Bessler \cite{bess-tech-repo-1960}, can only
be proved to be asymptotically optimal under a certain
assumption on the distributions
((\ref{eqn-SS-positivity-cond}) below).  We first show
that under the same assumption this test, which we
refer to as the {\em Chernoff} test,  is asymptotically
optimal in a {\em strong sense},
using the notion of decision making risk in
place of the overall probability of error.
\item We dispense with the aforementioned assumption by using a
modified version of the Chernoff test
described in Appendix \ref{app-B-2}, where  we outline the 
achievability proof of asymptotic optimality without
(\ref{eqn-SS-positivity-cond}).
\item We design another test to meet hard risk constraints while retaining asymptotic optimality.
\end{itemize}

Let $\mathcal{F}_k$ denote the $\sigma$-field generated
by $\left(Y^k, U^k \right)$.  A sequential test $\gamma
= \left( \phi, N, \delta \right)$ consists of a causal
observation control policy $\phi,$ an
$\mathcal{F}_k$-stopping time $N$ representing the
(random) number of observations before the final
decision, and the decision rule $\delta = \delta \left(
Y^N, U^N \right)$.  Akin to the paragraph containing
(\ref{eqn:FSS-CSdistribution-hypo-i}), the causal
control policy $\phi$ is described by the pmfs $q
\left( u_1 \right), \left \{ q \left( u_k \vert
y^{k-1}, u^{k-1} \right) \right \}_{k = 2}^{\infty}.$

\subsection{The Chernoff Test}
\label{subsection-SS-test}

We first present the Chernoff test
\cite{cher-amstat-1959,bess-tech-repo-1960} for
sequential design of experiments with multiple
hypotheses.  The proof of asymptotic optimality of this
test requires the following technical
assumption which was also imposed in
\cite{cher-amstat-1959, bess-tech-repo-1960}: {\em For
every} $u \in \mathcal{U},~0 \leq i < j < M-1,$
\begin{align}
D \left( p_i^u \| p_j^u \right) ~>~ 0.
\label{eqn-SS-positivity-cond}
\end{align} 
\indent
The  Chernoff test admits the following  sequential
description.  Having fixed the control policy up to
time $k$ and obtained the first $k$ observations and
control values $y^k, u^k$, if the controller decides to
continue taking more observations, then at time $k+1,$
a {\em randomized} control policy is adopted wherein
$U_{k+1} \in \mathcal{U}$ is drawn from the following
distribution
\begin{align}
    q \left( u \right) \,=\, q \left( u \Big \vert \hat{i}_k \right)
    \,=\,  \mathop{\mbox{argmax}}_{\overline{q} \left(u \right)}
    \min_{j \in \mathcal{M} \backslash \left \{ \hat{i}_k \right \} }
    \sum_u \overline{q} \left( u \right)
    D \left( p^u_{\hat{i}_k}  \| p^u_j  \right),
    \label{eqn-SS-opt-random-control}
\end{align}
where $\hat{i}_k \,=\, \mathop{\mbox{argmax}}_{i \in
\mathcal{M}}~ p_i \left ( y^k, u^k \right )$, is the ML
estimate of the hypothesis at time $k$. The stopping 
rule is defined as the first time $n$ for which
\begin{align}
    \log{ \left(
    \frac{
        p_{\hat{i}_n} \left( y^n, u^n \right)
    }
    {
        \max\limits_{j \,\neq\, \hat{i}_n} ~p_{j} \left( y^n, u^n \right)
    } \right)
    }
    ~\geq~
    - \log(c),
    \label{eqn-SS-stopping-rule}
\end{align}
where $c$ is a positive real-valued parameter that
will be selected to approach zero in order to drive the
probabilities of error to zero.
At the stopping time $n,$ the decision rule is ML,
i.e., $\delta \left( y^n, u^n \right) \,=\, \hat{i}_n.$
Note that randomization is used in the causal
control policy. This facilitates the
simultaneous minimization of the expected stopping time under the
$M$ hypotheses as the error probability goes to zero.
Also similar to the test proposed in the 
fixed sample size
setting, this sequential test relies on the separation
principle between estimation and control, with the
distinction that the stationary mapping from the
posterior distribution of the hypothesis to the control
value is now randomized.

To dispense with (\ref{eqn-SS-positivity-cond}), we 
propose a ``modified Chernoff test'' with a control 
policy that is slightly different from (\ref{eqn-SS-opt-random-control}).
Specifically, instead of using the policy (\ref{eqn-SS-opt-random-control})
at all times, we will occasionally sample from the uniform control 
independently of the index of the ML hypothesis; the specific way 
in which this is done will be explained in Appendix \ref{app-B-2}.
The stopping rule of this modify test 
will still be as in (\ref{eqn-SS-stopping-rule}) with 
the same $c$ therein.

\subsection{Asymptotic Optimality}
\label{subsection-SS-optimality}

In order to present a formal statement
establishing the strong asymptotic
optimality of the Chernoff test, we introduce the
concept of decision risks or frequentist error
probabilities \cite{drag-et-al-ieeetit-1999}.  In
particular, let $\pi ( i ),~i \in \mathcal{M},$ be a
prior distribution of the hypothesis with a full
support.  For each $i \in \mathcal{M},$ the {\em
probability of incorrectly deciding} $i$ or {\em the
risk of deciding} $i$ is given by
\begin{align}
    R_i ~\triangleq
    \sum_{j \in \mathcal{M} \backslash \left \{ i \right \}}
    \pi ( j ) ~\mathbb{P}_j \left \{ \delta = i \right \}.
    \label{eqn-SS-def-risks}
\end{align}
Note that for each $i \in \mathcal{M},$
\begin{align}
R_i = \sum_{j \in \mathcal{M} \backslash \left \{ i
\right \} }
    \pi ( j ) ~\mathbb{P}_j \left \{ \delta = i \right \}
~\leq~ \max\limits_{ k \in \mathcal{M} } ~\mathbb{P}_k
\left \{ \delta\ne k \right \},
\label{eqn-SS-risk-bounded-by-max-error}
\end{align}
Therefore, the condition $\max\limits_{ k \in
\mathcal{M} }~\mathbb{P}_k \left \{ \delta\ne k \right
\} \rightarrow 0$ implies that $\max\limits_{k \in
\mathcal{M}} R_k \to 0$. 

\vspace{0.03in}
\begin{theorem}
\label{Thm6}  The modified Chernoff test (as
$c \rightarrow 0$) satisfies
\begin{align}
\lim_{c \rightarrow 0}~ \max\limits_{i \in
\mathcal{M}}\, \mathbb{P}_i \left \{ \delta(Y^N,U^N)\ne
i \right \} ~=~ 0, \label{eqn-SS-error-prob-CS}
\end{align}
and for each $i \in \mathcal{M},$
\begin{align}
\hspace{-0.15in}
\mathbb{E}_i[N] &\leq \frac{ -\log{\left(
\max\limits_{k \in \mathcal{M}}  ~\mathbb{P}_k \left \{
\delta\ne k \right \} \right)} }
{  \max\limits_{q (u)} \,\min\limits_{ j \in{\cal M} \backslash \left \{ i \right \} } \,\sum\limits_{u } ~q(u)D(p^u_i \| p^u_j) }  \Big ( 1 + o (1) \Big ) 
\label{eqn-SS-asymp-opt-stoppingtime-CS}\\
&
\leq 
\frac{ -\log\left(  R_i \right)} 
{ \max\limits_{q (u)} \,\min\limits_{ j \in{\cal M}
\backslash \left \{ i \right \} } \,\sum\limits_{u }
~q(u)D(p^u_i \| p^u_j) } \Big ( 1 + o (1) \Big ).
\label{eqn-SS-asymp-opt-stoppingtime-CS-3}
\end{align}
\noindent
Furthermore, the modified Chernoff test is asymptotically
optimal in the following {\em strong sense}. If the
prior $\pi$ has full support on $\mathcal{M}$, then any
sequence of tests with vanishing maximal risk, i.e,
$\max\limits_{k \in \mathcal{M}} R_k \to 0$,
satisfies
for {\em
every} $i \in \mathcal{M}$,
\begin{equation}
\mathbb{E}_i[N]  ~\geq~ \frac{ -\log{\left( R_i
\right)} } { \max\limits_{q (u)} \,\min\limits_{ j
\in{\cal M} \backslash \left \{ i \right \} }
\,\sum\limits_{u } ~q(u)D(p^u_i \| p^u_j) }  \Big ( 1 +
o (1) \Big ).
\label{eqn-SS-asymp-opt-stoppingtime-CS-risk}
\end{equation}
\end{theorem}
\indent
\begin{remark}
The converse assertion
(\ref{eqn-SS-asymp-opt-stoppingtime-CS-risk}) in terms
of maximal risk implies the one in terms of the maximal
error probability, but not vice versa. Thus the
asymptotic optimality of the modified Chernoff test established
in Theorem~\ref{Thm6} is stronger than the
corresponding result in \cite{cher-amstat-1959,
bess-tech-repo-1960}, which is given in terms of
maximal error probability.
\end{remark}

\subsection{Asymptotically Optimal Test Meeting Hard Constraints on the Risks}
\label{subsection-SS-hard-risk-contraints} 

Although the calculation of risks involves
the prior distribution of the hypothesis, the test
proposed in Section \ref{subsection-SS-test} does not
use the knowledge of the prior distribution at all.  In
this section, we show that by using this knowledge, we
can further modify our test to meet hard constraints on the
risks. Another key to this new test is the use of
different thresholds for the peak of the posterior
distribution depending on the index of the
ML hypothesis
instead of a single threshold as in
(\ref{eqn-SS-stopping-rule}).  In the asymptotic regime
in which all the risks vanish, we show that this
modified test will
also be asymptotically optimal.
\newline \indent
Specifically, for a given tuple $\left(
\bar{R}_1, \ldots, \bar{R}_M \right),$ we will design a
test to satisfy $R_i \leq \bar{R}_i,\ i \in
\mathcal{M}$. To this end,
we modify the stopping rule
(\ref{eqn-SS-stopping-rule}) to be so that we stop at
the first time $n$ when 
\begin{align}
    \log{ \left(
    \frac{
       \pi \left( \hat{i}_n \right) ~p_{\hat{i}_n} \left( y^n, u^n \right)
    }
    {
        \max\limits_{j \,\neq\, \hat{i}_n} ~\pi \left( j \right) ~p_{j} \left( y^n, u^n \right)
    } \right)
    }
    ~\geq~
    \log{\left ( \frac{(M-1)\pi \left( \hat{i}_n \right)}{\bar{R}_{\hat{i}_n}} \right)}.
    \label{eqn-SS-stopping-rule-hard-risk-constraints}
\end{align}
\begin{theorem}
\label{Thm7}  For any tuple $\left( \bar{R}_1, \ldots,
\bar{R}_M \right), \bar{R}_i > 0,\ i \in \mathcal{M}$
and any $\pi$ with a full support, the modified
Chernoff test but with the stopping rule (\ref{eqn-SS-stopping-rule-hard-risk-constraints})
in place of (\ref{eqn-SS-stopping-rule})
satisfies, for every $i \in \mathcal{M}$,
\begin{equation}
    \sum_{j \neq i} \pi ( j )~
        \mathbb{P}_j \left \{ \delta \left( Y^N, U^N \right) = i \right \}
    ~\leq~ \bar{R}_i.
\label{eqn-Thm7-hard-risk-constraints}
\end{equation}
Furthermore,
as $\max\limits_{i \in \mathcal{M}} \bar{R}_i  
\hspace{-0.05in}
\rightarrow 
\hspace{-0.05in}
0,$ 
while satisfying $\max\limits_{i \in \mathcal{M}} \bar{R}_i \leq 
K \left( \min\limits_{i \in \mathcal{M}} \bar{R}_i \right)$ for some $K > 0$,
the proposed test
is asymptotically optimal, i.e., it
satisfies (\ref{eqn-SS-asymp-opt-stoppingtime-CS}) and, hence, 
also (\ref{eqn-SS-asymp-opt-stoppingtime-CS-3}).
\end{theorem}

\section{Example}
\label{section-Example}

We consider an example with parameters $\,M = 3,\
\mathcal{Y} = \left \{ 0, 1 \right \},\ \mathcal{U} =
\left \{ a, b, c \right \}.$
For an arbitrary $\,\epsilon,\ 0 < \epsilon <
\frac{1}{2},\,$ denote by $\,p \left( y \right)\,$ and
$\,\overline{p} \left( y \right),\,$ the two pmfs on
$\,\mathcal{Y}\,$ for which $\,p \left( 1 \right) =
\epsilon\,$ and $\,\overline{p} \left( 1 \right) = 1 -
\epsilon,\,$ respectively. Then, consider the model for
controlled sensing for hypothesis testing in which the
pmfs $~p_i ^u,\ i \in \left \{ 0, 1, 2 \right \},\ u
\in \left \{ a, b, c \right \},~$ are assigned
according to Table 1.
\begin{center}
 \scalebox{0.9}{
\hspace{0.6in} \begin{tabular}{ l | c  | c | c |}
                & $u = a$ & $u=b$ & $u=c$ \\ \hline
    $i = 0$ & $p$  & $\overline{p}$     & $\overline{p}$     \\ \hline
    $i = 1$ & $\overline{p}$       & $p$     & $\overline{p}$     \\ \hline
        $i = 2$ & $\overline{p}$       & $\overline{p}$     & $p$     \\ \hline
\end{tabular}}\newline

\small{Table 1:~Example}
\end{center}
This example is motivated by adaptive sensor selection
for event detection.  Consider a sensor network with a
fusion center and three sensors $a$, $b$ and $c$,
collecting measurements from three separate locations
0, 1, 2. A specific event takes place at exactly one
unknown location; it affects the distribution of the
measurements at this particular location (represented
by the distribution $p$ in Table 1), while the
measurements at the other two locations are distributed
according to $\overline{p}$.  At every time step, the
fusion center can query only one sensor to measure its
readings.  The goal is to determine the location of the
event in the most efficient manner.

The optimal exponent for open-loop control (cf.
(\ref{eqn:FSS-multiplehypo-opt-OLexponents})) can be
easily calculated to be
\begin{equation}
    \betaOL ~=~ \frac{2}{3} C \left( p , \overline{p} \right)
    ~=~ - \frac{2}{3}
    \log{ \left( 2 \sqrt{\epsilon \left( 1 - \epsilon \right)} \right)}.
\label{eqn-example1-OL-exponent}
\end{equation}

For causal control, we apply the control policy
presented in Section \ref{subsubsec-FSS-CS} (cf.
(\ref{eqn-FSS-opt-control})). Then, by solving the
maximization in (\ref{eqn-FSS-opt-control}), we obtain
a {\em deterministic} causal control policy, which is
given by $~u_{k+1} = f \left( \hat{i}_k \right),$ where
$f \left( 0 \right) = a,\
    f \left( 1 \right) = b,\
    f \left( 2 \right) = c$.\,
Lastly, at time $\,n,\,$ the decision is made for the
maximum likelihood estimate, i.e., $~\delta \left( y^n
\right) \,=\, \hat{i}_n$.
We now analyze the maximal error probability of this
test. To this end, for any $y^n,$ we let
\begin{align}
 \hspace{-0.1in}
 k_a
 \,=\,
 \Bigg
 \vert
 \left \{
    k \in \left \{ 1, \ldots, n \right \}
    \,:  \begin{array}{ll}
            	u_k = a ~\mbox{and}~y_k = 1,\ \mbox{or}
            \\
            	u_k \neq a~\mbox{and}~y_k = 0
            \end{array}
 \right \}
 \Bigg
 \vert .
\label{eqn-example1-number-bad-votes}
\end{align}
Then, we get from Table 1 that $~p_0 \left (
y^n \right )
    = \epsilon^{k_a} \left( 1 - \epsilon \right)^{n - k_a}.$
Similarly, we can define $k_b$ and $k_c$ with $a$ in
(\ref{eqn-example1-number-bad-votes}) replaced by $b$
and $c,$ respectively, and get that $p_1 \left ( y^n
\right ) =
    \epsilon^{k_b} \left( 1 - \epsilon \right)^{n - k_b},\
p_2 \left ( y^n \right ) =
    \epsilon^{k_c} \left( 1 - \epsilon \right)^{n - k_c}.
$

We sort $\,\left \{ k_a, k_b, k_c \right \}\,$ in an
ascending order and denote the sorted values by $~k_1
\leq k_2 \leq k_3.~$ Note that at every time step, the
most likely hypothesis is the one associated with
$k_1.$ Then, it follows from  Table 1 that as $n$
increases by one, if $\,y_n = 1,\,$ then \emph{the
least} of $\,\left \{ k_a, k_b, k_c \right \}\,$
increases by one, while the other two remain fixed.  On
the other hand, if $\,y_n = 0,\,$ then \emph{the least}
of $\,\left \{ k_a, k_b, k_c \right \}\,$ remains
fixed, while the other two increase by one. Hence, If
we let $\,k\,$ denote the number of zeros in $y^n$,
then $\frac{k_a + k_b + k_c}{3} ~=~ \frac{n+k}{3}$.
In addition, starting from no observation at time zero
when $\,\left \{ k_a, k_b, k_c \right \}\,$ are all
equal to zero, we get from an induction argument that,
$k_2 \leq k_3 \leq k_2 + 1.$
This argument is similar to that in [pp.
54]\cite{berl-phd-thes-1964}; we refer the reader to
\cite{berl-phd-thes-1964} for further details.
We can now conclude from these previous identities that
\begin{equation}
    k_2 ~\geq~
    \frac{k_a + k_b + k_c}{3} - \frac{1}{3} ~=~ \frac{n+k}{3}
    - \frac{1}{3}.
    \label{eqn-example1-bound-second-largest-vote}
\end{equation}
At time $n, \delta \left( y^n \right)\,$ corresponds to
the smallest $k_1;$ it follows from
(\ref{eqn-example1-bound-second-largest-vote}) that for
any $i = 0, 1, 2,$
\begin{eqnarray}
    \mathbb{P}_i \left \{ \delta \neq i \right \}
    &\leq& \sum_{y^n} \epsilon^{k_2 \left( y^n \right)}
                                \left( 1-\epsilon \right)^{n - k_2 \left( y^n \right)}
                                \nonumber \\
    &=& \sum_{w=1}^n~
    \sum_{y^n:\,\vert \left \{k:\,y_k = 0 \right \} \vert = w}
    \epsilon^{k_2 \left( y^n \right)}
                                \left( 1-\epsilon \right)^{n - k_2 \left( y^n \right)}
                                \nonumber \\
    &\leq&
    \left(
    \sum_{w=0}^n
    {\small \left(
        \begin{array}{cc}
            n \\
            w
        \end{array}
    \right) }
    \epsilon^{\frac{\left( n + w \right)}{3} - \frac{1}{3}}
    \left( 1- \epsilon \right)^{\frac{\left( 2n - w \right)}{3} + \frac{1}{3}}
    \right)
    \nonumber \\
    &=&
    \frac{
    \left(
    \epsilon^{\frac{1}{3}} \left( 1- \epsilon \right)^{\frac{2}{3}}
    +
    \epsilon^{\frac{2}{3}} \left( 1- \epsilon \right)^{\frac{1}{3}}
    \right)^n
    }
    {\epsilon^{\frac{1}{3}} \left( 1- \epsilon
    \right)^{-\frac{1}{3}}},
    \nonumber
\end{eqnarray}
and we get that
\begin{align}
\mathop{\overline{\mbox{lim}}}_{n \rightarrow \infty}\
- \frac{1}{n} \log{ \left( \max_{i = 0, 1, 2}
\mathbb{P}_i \left \{ \delta \neq i \right \} \right) }   \hspace{0.7in}  \nonumber \\
\ \geq\  -\log{ \left( \epsilon^{\frac{1}{3}} \left( 1-
\epsilon \right)^{\frac{2}{3}}
    +
    \epsilon^{\frac{2}{3}} \left( 1- \epsilon \right)^{\frac{1}{3}}
\right)}. \label{eqn-example1-CS-exponent}
\end{align}
Comparing (\ref{eqn-example1-OL-exponent}) to
(\ref{eqn-example1-CS-exponent}), we get that for every
$~\epsilon,\ 0 < \epsilon < \frac{1}{2},~$ causal
control can yield a larger error exponent than the best
open-loop control.
By the symmetry in Table 1, the upper bound for
$\betaC$ in (\ref{eqn:FSS-multiplehypo-opt-exponents})
can be calculated to be $C \left( p , \overline{p}
\right) = - \log{\left( 2 \sqrt{\epsilon \left( 1 -
\epsilon \right)}\right)}.$

Lastly, we show that our lower bound for $\betaC$ in
(\ref{eqn:FSS-multiplehypo-opt-exponents}) gives the
same achievable error exponent in
(\ref{eqn-example1-CS-exponent}) for this example. To
this end, we consider the argument of the $-\log{}$ in
the lower bound
\begin{equation}
    \sup_{\nu}
    \min_{u}
    \max_i
    \left(
    \sum_{y} p_i^u \left( y \right)
    e^{
    \frac{
    \eta \left [
                \nu \circ p^u \left( y \right) - p_i^u \left (y \right)
                \right ]
    }
    {\left( 1 - \nu \left( i \right) \right)}
    }
    \right).
    \label{eqn-Example1-lowerbd-betaC-1}
\end{equation}
Note that the argument minimizer $u$ in
(\ref{eqn-Example1-lowerbd-betaC-1}) is a function of
$\nu$. Hence, if the minimizer is replaced by a
specific function $u = f \left( \nu \right)$, then we
will get a larger quantity, i.e.,
\begin{align}
\sup_{\nu}
    \min_{u}
    \max_i
    \left(
    \sum_{y} p_i^u \left( y \right)
    e^{
    \frac{
    \eta (
                \nu \circ p^u \left( y \right) - p_i^u \left (y \right)
            )
    }
    {\left( 1 - \nu \left( i \right) \right)}
    }
    \right) \hspace{0.60in}  \nonumber \\
    \leq\    \sup_{\nu}
    \max_i
    \left(
    \sum_{y} p_i^{f \left( \nu \right)} \left( y \right)
    e^{
    \frac{
    \eta (
                \nu \circ p^{f \left( \nu \right)} \left( y \right)
            - p_i^{f \left( \nu \right)} \left (y \right)
            )
    }
    {\left( 1 - \nu \left( i \right) \right)}
    }
    \right).
    \hspace{0.1in}
        \label{eqn-Example1-lowerbd-betaC-2}
\end{align}
In particular, consider the following function
\begin{align}
    u ~=~ f \left( \nu \right)
    ~=~ \small \left \{
    \begin{array}{cc}
    a, & \mathop{\mbox{argmax}}_{i} \nu \left(i \right) = 0, \\
    b, & \mathop{\mbox{argmax}}_{i} \nu \left(i \right) = 1,\\
    c, & \mathop{\mbox{argmax}}_{i} \nu \left(i \right) = 2.
    \end{array}
    \right.
\label{eqn-Example1-lowerbd-betaC-3}
\end{align}

For an arbitrary $\nu \left( i \right),\ i = 0, 1, 2,$
denote their respective sorted values by $\nu_{\tiny u}
\geq \nu_{\tiny c} \geq \nu_{\tiny \ell}$.  Then, it
follows from (\ref{eqn-Example1-lowerbd-betaC-3})
via appropriate algebraic manipulations using Table 1
that
\begin{align}
\hspace{-0.2in} \sup_{\nu} \max_i
    \left(
    \sum_{y} p_i^{f \left( \nu \right)} \left( y \right)
    e^{
    \frac{
    \eta (
                \nu \circ p^{f \left( \nu \right)} \left( y \right)
            - p_i^{f \left( \nu \right)} \left (y \right)
            )
    }
    {\left( 1 - \nu \left( i \right) \right)}
    }
    \right)
   \hspace{0.45in} \nonumber \\
=\  \hspace{-0.1in} \sup_{ 1> \nu_{\tiny u} \geq
\nu_{\tiny c} \geq \nu_{\tiny \ell} } \,\hspace{-0.2in}
\max \left ( \hspace{-0.1in}
\begin{array}{cc}
    \left( 1 - \epsilon \right)
    e^{
    - \eta \left( 1 - 2 \epsilon \right)
    }
       \hspace{-0.04in}
    \,+\,
       \hspace{-0.04in}
    \epsilon~
    e^{
    \eta \left( 1 - 2 \epsilon \right)
    }, \\
    \left( 1 - \epsilon \right)
    e^{
    \frac{
    - \eta \left( 1 - 2 \epsilon \right)
    \nu_{\tiny u}
    }
    { \left( 1 - \nu_{\tiny c} \right)}
    }
   \hspace{-0.04in}
    \,+\,
       \hspace{-0.04in}
    \epsilon~
    e^{
    \frac{
    \eta \left( 1 - 2 \epsilon \right)
    \nu_{\tiny u}
    }
    { \left( 1 - \nu_{\tiny c} \right) }
    }, \\
   \left( 1 - \epsilon \right)
    e^{
    \frac{
    - \eta \left( 1 - 2 \epsilon \right)
    \nu_{\tiny u}
    }
    { \left( 1 - \nu_{\tiny \ell} \right) }
    }
       \hspace{-0.04in}
   \,+\,
       \hspace{-0.04in}
    \epsilon~
    e^{
    \frac{
    \eta \left( 1 - 2 \epsilon \right)
    \nu_{\tiny u}
    }
    { \left( 1 - \nu_{\tiny \ell} \right) }
    }
\end{array}
\hspace{-0.1in} \right ).
\label{eqn-Example1-lowerbd-betaC-5}
\end{align}
Next, we select $\eta = \frac{  2 \log{\left(
\frac{\left( 1-\epsilon \right)}{\epsilon} \right)} }
{3 \left( 1 - 2 \epsilon \right) }.
$
Note that for any $\nu_{\tiny u} \geq \nu_{\tiny c}
\geq \nu_{\tiny \ell},$
\begin{equation}
 \frac{1}{3} \ \leq\ \frac{ 2 \nu_{\tiny u}}{3 \left( 1 - \nu_{\tiny \ell}
 \right)} \ \leq\
  \frac{ 2 \nu_{\tiny u}}{3 \left( 1 - \nu_{\tiny c} \right)}
 \ \leq\ \frac{2}{3}.
\label{eqn-Example1-lowerbd-betaC-7}
\end{equation}
It then follows from the selection of $\eta,$
(\ref{eqn-Example1-lowerbd-betaC-7}), and the fact that
for any $0 < \epsilon < \frac{1}{2}$,
\begin{align}
\hspace{-0.02in} 
\max_{\frac{1}{3} \leq s \leq \frac{2}{3}} 
\hspace{-0.04in} 
\left( 1 - \epsilon \right)^{1-s}
\epsilon^s 
\hspace{-0.03in} 
+ 
\hspace{-0.03in} 
\left( 1 - \epsilon \right)^{s}
\epsilon^{1-s} 
\hspace{-0.025in} 
= 
\hspace{-0.025in} 
\left( 1 - \epsilon
\right)^{\frac{2}{3}} \epsilon^{\frac{1}{3}} 
\hspace{-0.03in} 
+ 
\hspace{-0.03in} 
\left( 1 - \epsilon \right)^{\frac{1}{3}}
\epsilon^{\frac{2}{3}}, \nonumber
\end{align}
that for any $\nu_{\tiny u} \geq \nu_{\tiny c} \geq
\nu_{\tiny \ell},$

\begin{align}
\max \left (
\begin{array}{cc}
    \left( 1 - \epsilon \right)
    e^{
    - \gamma
    }
    + \epsilon~
    e^{
        \gamma
    },\\
    \left( 1 - \epsilon \right)
    e^{
    - \gamma_{\tiny c}
    } +
    \epsilon~
    e^{
    \gamma_{\tiny c}
    },\\
   \left( 1 - \epsilon \right)
    e^{
    - \gamma_{\tiny l}
    } +
    \epsilon~
    e^{
    \gamma_{\tiny l}
    }
\end{array}
\right )
=  \left( 1 - \epsilon \right)^{\frac{2}{3}}
\epsilon^{\frac{1}{3}} + \left( 1 - \epsilon
\right)^{\frac{1}{3}} \epsilon^{\frac{2}{3}},
\nonumber
\end{align}
where $\gamma = \frac{  2 \log{\left( \frac{\left(
1-\epsilon \right)}{\epsilon} \right)} } {3},\
{\gamma}_{\tiny c} = \frac{  2 \log{\left( \frac{\left(
1-\epsilon \right)}{\epsilon} \right)} \nu_{\tiny u}}
{3 \left( 1 - \nu_{\tiny c} \right)},\ {\gamma}_{\tiny
l} = \frac{  2 \log{\left( \frac{\left( 1-\epsilon
\right)}{\epsilon} \right)} \nu_{\tiny u}} {3 \left( 1
- \nu_{\tiny l} \right)}.$ Following from
(\ref{eqn-Example1-lowerbd-betaC-2}) and
(\ref{eqn-Example1-lowerbd-betaC-5}) by 

\noindent
taking $-
\log{},$ we get that
\begin{align}
\betaC  &\ge - \log{ \left( \sup_{\nu} \max_i
    \left (
    \sum_{y} p_i^{f \left( \nu \right)} \left( y \right)
    e^{
    \frac{
    \eta (
                \nu \circ p^{f \left( \nu \right)} \left( y \right)
            - p_i^{f \left( \nu \right)} \left (y \right)
                )
    }
    {\left( 1 - \nu \left( i \right) \right)}
    }
    \right )
    \right)
}
\nonumber \\
&= -\log{ \left( \left( 1 - \epsilon
\right)^{\frac{2}{3}} \epsilon^{\frac{1}{3}} + \left( 1
- \epsilon \right)^{\frac{1}{3}} \epsilon^{\frac{2}{3}}
\right) } \nonumber
\end{align}
as required. This lower bound matches the one
in (\ref{eqn-example1-CS-exponent}). 

In the sequential setting, the quantities
dictating the asymptotically optimal performance are
$\max\limits_{q (u)} \, \min\limits_{ j \in{\cal M}
\backslash \left \{ i \right \} }\, \sum\limits_{u }
~q(u)D \left( p^u_i \| p^u_j \right),$ the denominators
on the right-side of
(\ref{eqn-SS-asymp-opt-stoppingtime-CS}), which can
readily be computed for this example to be
$-\log{\left( 2 \sqrt{\epsilon \left(1-\epsilon
\right)} \right)}$ for every $i \in \mathcal{M}$.  The
numerical value of this quantity is, as expected,
larger than $\betaC$ in the fixed sample size setting,
as now the control has an additional capability to
adaptively stop taking observations based on past
observations. 

\section{Discussion}
\label{section-Discussion}

In the proposed sequential test in Section
\ref{section-SS}, information from the past is used to
form the maximum likelihood estimate of the hypothesis,
which is used in turn to select the maximizing
distribution and the maximizing control value in
(\ref{eqn-SS-opt-random-control}). In contrast to
binary hypothesis testing with a fixed sample size (cf.
Proposition \ref{prop-1}), information from the past
seems to be crucial for attaining the asymptotically
optimal performance in the sequential setting, since
the mentioned maximizers can depend on the identity of
the ML hypothesis even for the case of binary
hypothesis testing.

\section{Conclusions}
\label{section-Conclusion} We studied the structure of
the optimal controller for multihypothesis testing with
observation control under various asymptotic regimes.
First, in a setting with a fixed sample size, 
the optimal error exponent corresponds to the maximum
Chernoff information over the choice of controls for
binary hypothesis testing.  In particular, in this
setup, a pure stationary open-loop control policy is
asymptotically optimal even among the broader class of
causal control policies.  For multiple hypothesis
testing, we characterized the optimal error exponent
achievable by open-loop control and derived precise
lower and upper bounds for the optimal error exponent
achievable by causal control.  We also proposed a
causal control policy for multihypothesis testing based
on maximizing the minimum Chernoff information of the
distributions corresponding to the most likely
hypothesis and all the alternative hypotheses. We
illustrated through an example that the proposed causal
control policy strictly outperforms the best open-loop
control policy.

Second, we considered a sequential setting wherein the
objective is to minimize the expected stopping time
subject to the constraints of vanishing error
probabilities under each hypothesis.  
We proposed 
a suitably modified version of the Chernoff test for multiple hypotheses testing and showed that it is asymptotically optimal 
in a strong sense, using the notion of decision making risk 
instead of the overall probability of error.  Our
control policy is based on maximizing the
KL distance
of the distributions corresponding to the most likely
hypothesis and the nearest alternative hypothesis.
We also designed another sequential test to
meet hard constraints on the risks while retaining the
asymptotic optimality.

For binary
hypothesis testing, the findings showed
that past information is crucial in achieving the
asymptotically optimal performance in the sequential
setting, while it is superfluous in the fixed sample
size setting. 
Our results also showed that
for general multiple hypothesis testing, randomization
in control is always superfluous (for any number of
hypotheses) in achieving the asymptotically optimal
performance in the fixed sample size setting. 
On the other hand, we showed that in the sequential setting, 
randomization can facilitate the structure of the asymptotically 
optimal control policy following the separation principle
between estimation and control especially in the
sequential setting.

\appendix
\subsection{Proof of Results in Section \ref{section-FSS}}
\label{app-A}

\subsubsection{Proof of Theorem \ref{Thm2}}
\label{app-A-1}

We start with the achievability proof.  First, note
that for any $n,$ and any test
\begin{align}
   \frac{1}{M} \sum_{i \in \mathcal{M}}  \mathbb{P}_i \left \{ \delta  \neq i \right \}
    &\ \leq\  \max_{i \in \mathcal{M}}~
    \mathbb{P}_i \left \{ \delta  \neq i \right \}
    \nonumber \\
    &\ \leq\ M ~\left(
    \frac{1}{M} \sum_{i \in \mathcal{M}}
    \mathbb{P}_i \left \{ \delta  \neq i \right \}
    \right).
    \label{eqn-Pf-Thm2-max-avg-error-relation}
\end{align}
Fix a sequence $u^n \in \mathcal{U}^n,$ and let
$\delta_{\text{ML}}: \mathcal{Y}^n \rightarrow
\mathcal{M}$ be the ML decision rule.
It now follows that
\begin{align}
\frac{1}{M} \sum_{i \in \mathcal{M}}
\mathbb{P}_i \left \{ \delta_{\text{ML}} \left(Y^n
\right) \neq i \right \} 	\hspace{1.5in} \nonumber \\
=\   \frac{1}{M} \sum_{i \in \mathcal{M}} ~\sum_{j
\in \mathcal{M} \backslash \left \{ i \right \}}
\mathbb{P}_i \left \{ \delta_{\text{ML}} \left(Y^n
\right) = j \right \}.
\label{eqn-Pf-Thm2-pairwise-perror}
\end{align}
For any $i, j,\ 0 \leq i < j \leq M-1,$ and any $s
\in [0, 1],$ we get that
\begin{align}
\mathbb{P}_i \left \{ \delta_{\text{ML}} \left(Y^n
\right) = j \right \} &\leq   \mathbb{P}_i \left \{
p_i \left ( y^n \right )^s p_j \left ( y^n \right
)^{1-s} \geq p_i \left (
y^n \right ) \right \},
\label{eqn-Pf-Thm2-pairwise-perror-1} 
\end{align}
and
\begin{align}
\hspace{-0.0in}
\mathbb{P}_j \left \{ \delta_{\text{ML}} \left(Y^n
\right) = i \right \} &\leq   \mathbb{P}_j \left \{
p_i \left ( y^n \right )^s p_j \left ( y^n \right
)^{1-s} \geq p_j \left ( y^n \right ) \right \}.
\label{eqn-Pf-Thm2-pairwise-perror-2}
\end{align}
Combining (\ref{eqn-Pf-Thm2-pairwise-perror-1}) and
(\ref{eqn-Pf-Thm2-pairwise-perror-2}), we obtain that
\begin{align}
\mathbb{P}_i \left \{ \delta_{\text{ML}}
\left(Y^n \right) = j \right \} + \mathbb{P}_j \left \{
\delta_{\text{ML}} \left(Y^n \right) = i \right \}
\hspace{1.0in}
\nonumber
\end{align}
\begin{align}
&\leq\  
\int_{y^n} \prod_{k=1}^n \left(
p_i^{u_k} \left( y_k \right)^s p_j^{u_k} \left( y_k
\right)^{1-s} \right) \prod_{k=1}^n d \mu_{u_k} \left(
y_k \right)
\nonumber \\
&=\  \prod_{k=1}^n \left( \int_{y_k} p_i^{u_k} \left(
y_k \right)^s p_j^{u_k} \left( y_k \right)^{1-s} d
\mu_{u_k} \left( y_k \right) \right)
\nonumber \\
&=\ e^{n \left( \sum\limits_{u \in \mathcal{U}}
\overline{q}(u) \log{\left( \int_{y} p_i^{u} \left( y
\right)^s p_j^{u} \left( y \right)^{1-s} d \mu_{u}
\left( y \right) \right)} \right) },
\label{eqn-Pf-Thm2-lowerbd-OLexponent}
\end{align}
where $\overline{q} \left( \cdot \right)$ denotes the
empirical distribution of $u^n:~ \overline{q} \left( u
\right) \ \triangleq\ \frac{1}{n} \vert
 \left \{k:\
    k \in \left \{ 1, \ldots, n \right \},\ u_k = u
 \right \}
\vert.$ Since (\ref{eqn-Pf-Thm2-lowerbd-OLexponent}) is
true for any $s \in [0, 1],$ we get that
\begin{align}
\mathbb{P}_i \left \{ \delta_{\text{ML}} \left(Y^n
\right) = j \right \} + \mathbb{P}_j \left \{
\delta_{\text{ML}} \left(Y^n \right) = i \right \}
\hspace{1.0in} \nonumber \\
\leq\  e^{-n \left( \max\limits_{s \in [0,1]}
~-\sum\limits_{u \in \mathcal{U}} \overline{q}(u)
\log{\left( \int_{y} p_i^{u} \left( y \right)^s p_j^{u}
\left( y \right)^{1-s} d \mu_u \left( y \right)
\right)} \right) }. \nonumber
\end{align}
Because there are only finitely many pairs of
hypotheses in the sum on the right-side of
(\ref{eqn-Pf-Thm2-pairwise-perror}), the pair
corresponding to the smallest exponent will dominate
the exponent.  Hence, we get
\begin{align}
\frac{1}{M} \sum_{i \in \mathcal{M}}
\mathbb{P}_i \left \{ \delta_{\text{ML}} \left(Y^n
\right) \neq i \right \} \hspace{2.04in} \nonumber \\
\leq  
\hspace{-0.02in}
\left( M-1 \right) e^{-n \left( \min\limits_{i
< j} \max\limits_{s \in [0,1]} -\sum\limits_{u}
\overline{q}(u) \log{\left( \int_{y} p_i^{u} \left( y
\right)^s p_j^{u} \left( y \right)^{1-s} d \mu_u \left(
y \right) \right)} \right) }.
\nonumber
\end{align}

Since $u^n$ is arbitrary, we can approximate any
distribution $q(u)$ arbitrarily close by the empirical
distribution $\overline{q}^{(n)}(u)$ of an appropriate
{\em deterministic} sequence $u^n$ such that $\max_{u}
\vert \overline{q}^{(n)}(u) - q(u) \vert \rightarrow
0$. This fact combined with
(\ref{eqn-Pf-Thm2-max-avg-error-relation}) yields that
\begin{align}
\betaOL \ \geq\ \hspace{3.05in} \nonumber \\
\max_{q(u)} \min_{i < j} \max_{s \in [0,1]} -\sum_{u}
q(u) \log{\left( \int_{y} p_i^{u} \left( y \right)^s
p_j^{u} \left( y \right)^{1-s} d \mu_u \left( y \right)
\right)}, \label{eqn-Pf-Thm2-lowerbd-OLexponent3}
\end{align}
and that the error exponent on the right-side of
(\ref{eqn-Pf-Thm2-lowerbd-OLexponent3}) is achievable
by pure open-loop control.

Next, we prove that the reverse inequality of
(\ref{eqn-Pf-Thm2-lowerbd-OLexponent3}).  Since we
proved that $\betaOL$ is achievable by pure control, we
restrict our attention to pure open-loop control. By
considering the necessary and sufficient condition for
the maximizing $s$ of the function
\begin{equation}
        - \sum_{k=1}^n
        \log{\left(
            \int_{y_k} p_i^{u_k} \left( y_k \right)^s
            p_j^{u_k} \left( y_k \right)^{1-s}
            d \mu_{u_k} \left( y_k \right)
        \right)},
\nonumber
\end{equation}
we obtain for any $u^n \in \mathcal{U}^n,$ and any pair
of hypotheses $i, j \in \mathcal{M}$ that
\begin{equation}
\hspace{-0.0in}
s^* = \mathop{\mbox{argmax}}_{s \in [0,1]}
        -\sum_{k=1}^n
        \log{\left(
            \int_{y_k} p_i^{u_k} \left( y_k \right)^s
            p_j^{u_k} \left( y_k \right)^{1-s}
            d \mu_{u_k}
        \right)}
\label{eqn-Pf-Thm2-argmax-s-OL}
\end{equation}
satisfies
(cf.(\ref{eqn-FSS-projected-distribution}))
\begin{align}
\max_{s \in \left [ 0, 1 \right ]}
    \, - \sum_{k=1}^n
    ~\log \left(
    \int_{y_k} p_i^{u_k} \left( y_k \right)^s
    p_j^{u_k} \left( y_k \right)^{1-s}
    \mu_{u_k} \left( y_k \right)
    \right)	\hspace{0.35in} \nonumber
\end{align}
\vspace{-0.17in}
\begin{align}
 	=~
    \sum_{k=1}^n
    \hspace{-0.02in}
    D \left( b^{u_k, s^*}_{ij} \| p_i^{u_k} \right)
    ~=~
    \sum_{k=1}^n
        \hspace{-0.02in}
    D \left( b^{u_k, s^*}_{ij} \| p_j^{u_k} \right).
    \label{eqn-Pf-Thm2-equalization}
\end{align}
We next consider, for the same pair of hypotheses $i,j$
as above, the pdf/pmf $\tilde{p}$ defined by $\tilde{p}
\left( y^n \right) \ \triangleq\
\prod_{k=1}^nb_{ij}^{u_k, s^*} \left( y_k \right). $
~For any test, it either holds that
\begin{equation}
    \tilde{\mathbb{P}}
    \left \{
        \delta \left( Y^n \right) = i
    \right \} \,\geq\,  \frac{1}{2},\ \mbox{or~that}\ \ 
    \tilde{\mathbb{P}}
    \left \{
        \delta \left( Y^n \right) \neq i
    \right \} \,\geq\,  \frac{1}{2}.
    \label{eqn-Pf-Thm2-conv-case}
\end{equation}

Suppose that the first case of
(\ref{eqn-Pf-Thm2-conv-case}) holds.  
For any causal
control policy, under the stationary Markovity
assumption and assumption (\ref{eqn-assumption-2}), it
follows that the random process $S_k,\ k = 1, \ldots,
n,$ where
\begin{align}
S_k 
    \hspace{-0.02in}	
\triangleq 
    \hspace{-0.02in}	
    \sum_{l = 1}^k 
    \hspace{-0.03in}	
    \left(
    \log{
            \left(
            \frac{b^{u_l, s^*}_{ij} \left( Y_l \right)}
            {p_j^{u_l}\left( Y_l \right)}
            \right)
    }
    \hspace{-0.035in}
    -
    \hspace{-0.015in}
    \mathbb{E} \left [
            \log{
            \left(
                \frac{b^{u_l, s^*}_{ij} \left( Y_l \right)}
            {p_j^{u_l}\left( Y_l \right)} \right)
         }
    \bigg \vert \mathcal{F}_{l-1}
    \right ]
\right), 
\nonumber
\end{align}
is a ``stable'' martingale adapted to $\mathcal{F}_k,$
the sigma fields generated by $\left( Y^k, U^k
\right),\ k = 1, \ldots, n.$  By the martingale
stability theorem of Lo\`eve \cite[pp.
53]{loev-prob-thy-II-book-1978},  we get that $\left \{
\frac{1}{n} S_n \right \}_{n=1}^{\infty}$ converges to
zero a.s. and, hence, in probability, i.e., for any
$\eta > 0,$
\begin{align}
\lim_{n \rightarrow \infty} \tilde{\mathbb{P}} \left \{
\frac{1}{n} \sum_{k=1}^n 
\left( 
\begin{array}{ll}
\log{\left(
\frac{b^{u_k, s^*}_{ij} \left( Y_k \right)}
          {p_j^{u_k} \left( Y_k \right)} \right)}  \ -\\
    \mathbb{E} \left [
     \log{\left( \frac{b^{u_k, s^*}_{ij} \left( Y_k \right)}
          {p_j^{u_k} \left( Y_k \right)} \right)}
          \bigg \vert \mathcal{F}_{k-1}
        \right ]
\end{array}
\right)
> \eta
\right \} =  0.
\label{eqn-Pf-Thm2-p-conv-martingale}
\end{align}
Since $u_k,\ k = 1, \ldots, n$ are fixed (pure
open-loop control policy), we obtain from
(\ref{eqn-Pf-Thm2-p-conv-martingale}) that
\begin{align}
\hspace{-0.05in}
\lim_{n \rightarrow \infty} \tilde{\mathbb{P}} \left \{
\frac{1}{n} \sum_{k=1}^n \left( 
\begin{array}{ll}
\log{\left(\frac{b^{u_k, s^*}_{ij} \left( Y_k \right)}
          {p_j^{u_k} \left( Y_k \right)} \right)} \\
     -~  D \left( b_{ij}^{u_k, s^*} \| p_j^{u_k} \right)
\end{array}
\right)
> \eta
\right \} = 0.
\label{eqn-Pf-Thm2-p-conv-martingale-2}
\end{align}
The first inequality of (\ref{eqn-Pf-Thm2-conv-case})
and (\ref{eqn-Pf-Thm2-p-conv-martingale-2}) yield that
for any $\epsilon' > 0,$ any $\eta > 0$ and all $n$
large,
\begin{align}
 \frac{1}{2} - \epsilon'
 \hspace{3.1in} \nonumber
\end{align}
\begin{align}
 &\leq~ \tilde{\mathbb{P}}
    \left \{
        \hspace{-0.01in}
        \begin{array}{ll}
        \delta \left( Y^n \right) \ =\  i, \\
        \prod\limits_{k=1}^n p_j^{u_k} \left(Y_k \right)
        \,>\,
        \begin{array}{ll}
        	  e^{-n \left(
            \sum\limits_{k=1}^n \frac{1}{n}
                d \left(
                D \left( b_{ij}^{u_k, s^*} \| p_j^{u_k} \right)
                \right)
                 + \eta
            \right)} \\
        \times \prod\limits_{k=1}^n b_{ij}^{u_k, s^*} \left( Y_k \right)
        \end{array}
        \end{array}
        \hspace{-0.01in}
    \right \}   \nonumber \\
    &\leq~
    \mathbb{P}_j \left \{ \delta \left( Y^n \right) \neq j \right \}
    e^{n \left(
            \sum\limits_{k=1}^n \frac{1}{n}
                D \left( b_{ij}^{u_k, s^*} \| p_j^{u_k} \right)
                 + \eta
            \right)}.
            \label{eqn-Pf-Thm2-upperbd-exponent-hypo-j}
\end{align}
If the second case of (\ref{eqn-Pf-Thm2-conv-case})
holds instead, then similar to
(\ref{eqn-Pf-Thm2-p-conv-martingale}) we obtain that
\begin{align}
\lim_{n \rightarrow \infty} ~\tilde{\mathbb{P}} \left \{
\frac{1}{n} \sum_{k=1}^n \left( 
\begin{array}{ll}
\log{\left(
\frac{b^{u_k, s^*}_{ij} \left( Y_k \right)}
          {p_i^{u_k} \left( Y_k \right)} \right)} \\
     -~ D\  \left( b_{ij}^{u_k, s^*} \| p_i^{u_k} \right) 
\end{array}
\right)
~>~ \eta
\right \} \ =\  0.
\label{eqn-Pf-Thm2-p-conv-martingale-3}
\end{align}
From the second case of (\ref{eqn-Pf-Thm2-conv-case})
and (\ref{eqn-Pf-Thm2-p-conv-martingale-3}), we obtain
that for any $\epsilon'
> 0$ and any $\eta > 0,$
\begin{align}
\hspace{-0.11in}
\frac{1}{2} - \epsilon'  
\leq   \mathbb{P}_i \left \{
\delta \left( Y^n \right) \neq i \right \}
    e^{n \left(
            \sum\limits_{k=1}^n \frac{1}{n}
                D \left( b_{ij}^{u_k, s^*} \| p_i^{u_k} \right)
                 + \eta
            \right)}, 
            \label{eqn-Pf-Thm2-upperbd-exponent-hypo-i}
\end{align}
which parallels
(\ref{eqn-Pf-Thm2-upperbd-exponent-hypo-j}). It now
follows from
(\ref{eqn-Pf-Thm2-upperbd-exponent-hypo-j}),
(\ref{eqn-Pf-Thm2-upperbd-exponent-hypo-i}) and
(\ref{eqn-Pf-Thm2-equalization}) that for any $i, j \in
\mathcal{M},$
\begin{align}
\lim_{n \rightarrow \infty}
-\frac{1}{n} \log{ \left( \max_{i \in \mathcal{M}}
~\mathbb{P}_i \left \{ \delta
\neq i\right \} \right) }
\hspace{1.6in}		\nonumber 
\end{align}
\begin{align}
&\leq~ \lim_{n \rightarrow \infty} -\frac{1}{n} \log{
\left( \mbox{max} \left(
\mathbb{P}_i \left \{ \delta \left( Y^n \right) \neq
i\right \},
\mathbb{P}_j \left \{ \delta \left( Y^n \right) \neq j
\right \}
 \right) \right)}  \nonumber \\
&\leq~ \mbox{max} \left(
 \frac{1}{n} \sum\limits_{k=1}^n
 D \left( b^{u_k, s^*}_{ij} \| p_i^{u_k} \right),
 \frac{1}{n} \sum\limits_{k=1}^n
 D \left( b^{u_k, s^*}_{ij} \| p_j^{u_k} \right)
\right)
\nonumber \\
&= \frac{1}{n} \sum_{k=1}^n
           -\log \left(
           \int_{y_k} p_i^{u_k} \left( y_k \right)^{s^*}
           p_j^{u_k} \left( y_k \right)^{1-s^*}
           d \mu_{u_k} \left( y_k \right)
           \right)
           \nonumber \\
&= -\sum_{u} \overline{q}(u)
           \log{ \left(
           \int_{y} p_i^{u} \left( y \right)^{s^*}
           p_j^{u} \left( y \right)^{1-s^*}
           \mu_u \left( y \right)
           \right) } \nonumber \\
&= \max_{s \in [0,1]} - \sum_{u} \overline{q}(u)
           \log{ \left(
           \int_{y} p_i^{u} \left( y \right)^{s}
           p_j^{u} \left( y \right)^{1-s}
           \mu_u \left( y \right)
           \right) },
\label{eqn-Pf-Thm2-upperbd-OLexponent-ij}
\end{align}
where $\overline{q}$ denotes the empirical
distribution of $u^n$. Since
(\ref{eqn-Pf-Thm2-upperbd-OLexponent-ij}) must hold for
every pair $i, j$ of hypotheses, we then get that
\begin{align}
\lim_{n \rightarrow \infty} -\frac{1}{n} \log{ \left(
\max_{i \in \mathcal{M}} ~\mathbb{P}_i \left \{ \delta
\left( Y^n \right) \neq i\right \} \right)}
\hspace{1.36in}
\nonumber \\
\leq  \min_{i < j}
 \max_{s \in [0,1]} - \sum_{u} \overline{q}(u)
           \log{ \left(
           \int_{y} p_i^{u} \left( y \right)^{s^*}
           p_j^{u} \left( y \right)^{1-s^*}
           d \mu_u \left( y \right)
           \right) },
\nonumber
\end{align}
and, hence,
\begin{align}
\betaOL \hspace{3.3in}  \nonumber \\
\leq
 \max_{q(u)}
 \min_{i < j}
 \max_{s} -\sum_{u} q(u)
           \log{ \left(
           \int_{y} p_i^{u} \left( y \right)^{s^*}
           p_j^{u} \left( y \right)^{1-s^*}
           d \mu_u
           \right) }.
           \label{eqn-Pf-Thm2-upperbd-OLexponent}
\end{align}
Note that in (\ref{eqn-Pf-Thm2-upperbd-OLexponent-ij}),
the empirical distribution $\overline{q}(u)$ depends
only on the pure control $u^n$ and {\em not} on the
pair of hypotheses $i, j$, while the maximizer $s^*$ in
(\ref{eqn-Pf-Thm2-argmax-s-OL}) depends {\em both} on
$u^n$ and on the pair of hypotheses. The assertion of
Theorem 2 is now proved by combining
(\ref{eqn-Pf-Thm2-lowerbd-OLexponent3}) and
(\ref{eqn-Pf-Thm2-upperbd-OLexponent}).

\vspace{0.1in}
\subsubsection{Proof of Theorem \ref{Thm3}}
\label{app-A-2}

We first prove that $\betaC$ is achievable by a pure
control policy. For any fixed $n$, the problem of
finding the optimal causal control that
minimizes the {\em exact} average probability of error
can be cast as a finite-horizon stochastic optimal
control problem through the use of the posterior
distribution as a sufficient statistic. Since
$\,\mathcal{U}\,$ is finite, it follows from a standard
dynamic programming argument
\cite{bert-dyn-prog-opt-cont-book-2000} that the
optimal causal control is a deterministic one.

Next, we prove the upper bound for $\betaC$ in
(\ref{eqn:FSS-multiplehypo-opt-exponents}).
Observe that for any test for $M$ hypotheses, with a
decision rule $\delta$
and any pair of hypotheses $i, j \in \mathcal{M},$ a
binary test
for hypotheses $i$ and $j,$  can be constructed
using the same control policy and an appropriate
decision rule $\tilde{\delta}$ so that
\begin{align}
\max \left( \mathbb{P}_i \left \{ \tilde{\delta} \left(
Y^n, U^n \right) \neq i \right \}, \mathbb{P}_j \left
\{ \tilde{\delta} \left( Y^n, U^n \right) \neq j \right
\} \right)
\hspace{0.2in}
\nonumber \\
\leq\ \  \max_{i \in \mathcal{M}} ~\mathbb{P}_i
\left \{ \delta \left( Y^n, U^n \right) \neq i \right
\}. \nonumber
\end{align}
Applying the converse part of Theorem 1 with the roles
of $\left \{ p_0^u \right \}_{u \in \mathcal{U}}$ and
$\left \{ p_1^u \right \}_{u \in \mathcal{U}}$ therein
being played by $\left \{ p_i^u \right \}_{u \in
\mathcal{U}}$ and $\left \{ p_j^u \right \}_{u \in
\mathcal{U}},$ respectively, we obtain that
\begin{equation}
\betaC ~\leq~
    \max_{u \in \mathcal{U}}
    ~\max_{s \in \left [ 0, 1 \right ]}
           ~-\log \left(
           \int_{y} p_i^u \left( x \right)^s
           p_j^u \left( y \right)^{1-s}
           d \mu_u \left( y \right)
           \right).
           \nonumber
\end{equation}
As the previous argument applies for any $i \neq j,\ i,
j \in \mathcal{M},$ we obtain the upper bound in
(\ref{eqn:FSS-multiplehypo-opt-exponents}) by
minimizing over all pairs of hypotheses $i, j \in
\mathcal{M}$.

It is then only left to prove the lower bound for
$\betaC$ in (\ref{eqn:FSS-multiplehypo-opt-exponents}).
The proof relies on the following lemma whose proof is
deferred to Appendix \ref{app-A-3}.

\begin{lemma}
\label{lemma-1}
Let $J = \vert \mathcal{Y} \vert$. For every $\epsilon,
0 < \epsilon < 1,$ and $\eta > 0,$ it holds that
\begin{align}
    \sup_{\nu}
    \min_{u}
    \max_i
    ~\sum_{y} p_i^u \left( y \right)
    \left (
    \frac{
    1 + \epsilon ( J\, p_i^u \left( y \right) - 1 )
    }
    {
    1 + \epsilon (
    \sum\limits_{ j \neq i }
     \frac{ J\, \nu \left( j \right)  p_j^u \left( y \right)}
     {\left( 1 - \nu \left( i \right) \right)}
     - 1 )
    }
    \right )^{ -\frac{\eta}{ J \epsilon }}
	\nonumber 
\end{align}
\begin{align}
\hspace{0.4in}
\geq\  \ e^{-\betaC}, \label{Pf-Thm3-lower-bound-1}
\end{align}
where the outer supremum on the left-side of
(\ref{Pf-Thm3-lower-bound-1}) is over the set of all
pmfs on $\mathcal{M}$ that are not point-mass
distributions.
\end{lemma}

By L'H\^{o}pital's rule, for every $\nu$ that is not a
point-mass distribution,
\begin{align}
\hspace{-0.05in}
\lim_{\epsilon \rightarrow 0} \left (
    \frac{
    1 + \epsilon ( J\, p_i^u \left( y \right) - 1 )
    }
    {
    1 + \epsilon (
    \sum\limits_{ j \neq i }
    \hspace{-0.05in}
     \frac{ J\, \nu \left( j \right)  p_j^u \left( y \right)}
     {\left( 1 - \nu \left( i \right) \right)}
     - 1 )
    }
    \right )^{ -\frac{\eta}{ J \epsilon }}
    \hspace{-0.25in}
&\rightarrow\,
   e^{
    \eta \left (
                \sum\limits_{ j \neq i }
            \hspace{-0.05in}
            \frac{\nu \left ( j \right) p_j^u \left( y \right)}
            {\left( 1 - \nu \left( i \right) \right)}
            \ -\  p_i^u \left (y \right)
               \right )
    }
    \nonumber \\
&=\,
   e^{
    \frac{
    \eta (
                \nu \circ p^u \left( y \right) - p_i^u \left (y \right)
            )
    }
    {\left( 1 - \nu \left( i \right) \right)}
    }.
\label{Pf-Thm3-lower-bound-3}
\end{align}
Consequently, by letting $\epsilon \rightarrow 0,$ we
get from Lemma \ref{lemma-1} and
(\ref{Pf-Thm3-lower-bound-3}) through the finiteness of
$\mathcal{M},~\mathcal{U},~ \mathcal{Y}$ that for any
$\eta
> 0$,
\begin{equation}
- \log{
    \left(
    \sup_{\nu}
    \min_{u}
    \max_i
    \left(
    \sum_{y} p_i^u \left( y \right)
    e^{
    \frac{
    \eta \left [
                \nu \circ p^u \left( y \right) - p_i^u \left (y \right)
                \right ]
    }
    {\left( 1 - \nu \left( i \right) \right)}
    }
    \right)
    \right)
    }
    \ \ \leq\ \  \betaC. \nonumber
\end{equation}
The proof of Theorem 3 follows by optimizing over
$\eta > 0$.

\subsubsection{Proof of Lemma \ref{lemma-1}}
\label{app-A-3}

We shall consider a test based on a mismatched
posterior distribution on the hypothesis.
In particular, the control value at every time is
picked based on the posterior distribution on
$\mathcal{M}$ computed based on an appropriately chosen
mismatched model $\left \{ q_i^u \right \}_{i \in
\mathcal{M}}^{ u \in \mathcal{U}}$ (instead of the real
model $\left \{ p_i^u \right \}_{i \in \mathcal{M}}^{ u
\in \mathcal{U}}$) and the uniform prior distribution
on $\mathcal{M}$.
In particular, denote the posterior probability of
hypothesis $i \in \mathcal{M}$ at time $k = 0, \ldots,
n,$ by $\nu_k \left( i \right)$. Then,
\begin{align}
    \nu_0 \left( i \right) \ =\ \frac{1}{M},\
    \nu_k \left( i \right)
    \ =\  \frac{ \prod\limits_{l=1}^k q_i^{u_l \left( y^{l-1} \right) } \left( y_l \right) }
    {\sum\limits_j \prod\limits_{l=1}^k q_j^{u_l \left( y^{l-1} \right) } \left( y_l \right)},
    \ 1 \leq k \leq n.
    \nonumber
\end{align}
Also denote the likelihood ratio for hypothesis $i \in
\mathcal{M}$ at time $k = 0, \ldots, n,$ by $l_k \left(
i \right)$, i.e.,
\begin{align}
    l_0 \left( i \right) \ =\  \frac{1}{M-1}, \ \ 
    l_k \left( i \right) \ \triangleq\ \frac{\nu_k  \left( i \right)}{1- \nu_k \left(i \right)}
    \ =\ \frac{\nu_k \left( i \right)}{\sum\limits_{j \neq i} \nu_k \left( j \right)},
    \hspace{0.1in}
    \nonumber \\
    l_{k+1} \left( i \right) \,=\,
    \frac{\nu_k \left( i \right) q_i^{u_{k+1}\left( \nu_k \left( y^k \right) \right)}
                \left( y_{k+1} \right)}
    { \sum\limits_{j \neq i}
    \nu_k \left( j \right) q_j^{u_{k+1}\left( \nu_k \left( y^k \right) \right)} \left( y_{k+1} \right)
    },\ 0 \leq k \leq n-1.
    \label{Pf-Lemma1-2}
\end{align}
The decision rule at time $n$ is the maximum likelihood
estimate of the hypothesis, i.e., $\,\delta \left( y^n
\right) =
    \mathop{\mbox{argmax}}_{i} ~\nu_n \left( i \right).\,$
Next, we analyze the probability of error of such test
as a
function of $\left \{ q_i^u \right \},$~
$i \in \mathcal{M},\ u \in \mathcal{U},$ and the pure
control $u_k = u_k \left( \nu_{k-1} \right) =  u_k
\left( y^{k-1} \right)$ which will be specified later.
We get that for any $\lambda < 0$, the probability of
error (with respect to the {\em real} model $\left \{
p_i^u \right \},\ i \in \mathcal{M},\ u \in
\mathcal{U}$) under hypothesis $i$ can be upper bounded
as
\begin{align}
\mathbb{P}_i \left \{ \delta \neq i \right \}
    &\ =\  \mathbb{P}_i \left \{
    \mathop{\mbox{argmax}}_{j} \Pi_n \left( j \right) \neq i
    \right \}		\nonumber \\
    &\ \leq\   \mathbb{P}_i \left \{
            L_n \left( i \right) \leq 1
    \right \}
    \ \leq\   \mathbb{E}_i \left [ L_n \left( i \right)^{\lambda} \right ].
    \label{Pf-Lemma1-4}
\end{align}
Next, by writing
\begin{align}
 L_n \left( i \right) \,=\,  \prod_{k=1}^n \left ( \frac{L_k \left( i \right)}
                                                {L_{k-1}\left( i \right)}
                                                \right)
                                                L_0 \left( i \right)
                     \,=\,  \prod_{k=1}^n \left ( \frac{L_k \left( i \right)}
                                                {L_{k-1}\left( i \right)} \right)
                                                \frac{1}{M-1},
    \label{Pf-Lemma1-5}
\end{align}
and substituting (\ref{Pf-Lemma1-5}) into
(\ref{Pf-Lemma1-4}), we get that for any $\lambda < 0,$
\begin{equation}
    \mathbb{P}_i \left \{
    \delta \neq i
    \right \}
    \ \leq\  \mathbb{E}_i \left [
    \prod_{k=1}^n \left ( \frac{L_k \left( i \right)}
                                                {L_{k-1}\left( i \right)} \right)^{\lambda}
    \right ] (M-1)^{- \lambda}.
    \label{Pf-Lemma1-6}
\end{equation}
We next specify the mismatched model $\left \{ q_i^u
\right \},\ i \in \mathcal{M},\ u \in \mathcal{U}$. For
any $\epsilon,\ 0 < \epsilon < 1$, consider the
conditional pmf $W_{\epsilon} \left( y \vert y'
\right),\ y, y' \in \mathcal{Y},$ such that
\begin{equation}
W_{\epsilon} \left( y \vert y' \right) ~=~ \left \{
\begin{array}{cc}
    \frac{1}{J} + \frac{J-1}{J} \epsilon, & y = y', \\
    \frac{1}{J} - \frac{1}{J} \epsilon, & y \neq y'.
\end{array}
\right.
\end{equation}
Then, let
\begin{align}
q_i^u \left( y \right)
 &~=~ p_i^u \circ W_{\epsilon} \left( y \right)
        ~\triangleq~ \sum\limits_{y'}
        p_i^u \left( y' \right) W_{\epsilon} \left( y \vert y' \right)
        \nonumber \\
 &~=~ p_i^u \left( y \right) \left( \frac{1}{J} + \frac{(J-1) \epsilon}{J}  \right)
         \,+\, \sum\limits_{y' \neq y}
         p_i^u \left( y' \right) \left( \frac{1}{J} - \frac{\epsilon}{J}  \right)
         \nonumber \\
 &~=~  \frac{1}{J} ~+~
        \frac{\epsilon}{J} \left( J p_i^u \left( y \right) - 1 \right).
            \label{Pf-Lemma1-7}
\end{align}
Using this particular $\left \{ q_i^u \right \},\ i \in
\mathcal{M},\ u \in \mathcal{U},$ with
$\mathcal{F}_{k-1}$ denoting the sigma field generated
by $y^{k-1},\ k = 1, \ldots, n$, we get from
(\ref{Pf-Lemma1-2}) through an easy algebraic
manipulation that
\begin{align}
\mathbb{E}_i \left [ \left( \frac{L_k \left( i \right)}
        {L_{k-1} \left( i \right)} \right)^{\lambda}
        \,\Big \vert \mathcal{F}_{k-1} \right ]
\hspace{1.9in} 	\nonumber
\end{align}
\vspace{-0.2in}
\begin{align}
&~=~ \sum_{y} p_i^{u_k} \left( y \right)
    \left (
    \frac
    {
    \left( 1 - \nu_{k-1}\left( i \right) \right) q_i^{u_k} \left( y \right)
    }
    {
    \sum\limits_{j \neq i} \nu_{k-1}\left( j \right) q_j^{u_k}\left( y \right)
    }
    \right )^{\lambda} \nonumber \\
&~=~
    \sum_{y} p_i^{u_k} \left( y \right)
    \left (
    \frac{
    1 ~+~ \epsilon \left [ J\, p_i^{u_k} \left( y \right) - 1 \right ]
    }
    {
    1 ~+~ \epsilon  \left (
    \sum\limits_{ j \neq i }
    \frac{
    J\, \nu_{k-1} \left( j \right)  p_j^{u_k} \left( y \right)
    }
    {
     \left( 1 - \nu_{k-1} \left( i \right) \right)
    } 
     \,-\, 1  \right )
    }
    \right )^{\lambda}, 
    \nonumber
\end{align}
where $u_k = u_k \left( \nu_{k-1} \right) = u_k \left(
y^{k-1} \right)$. Next, let $\lambda = -\frac{\eta}{J
\epsilon}$ for an arbitrary $\eta > 0$, and let
\vspace{-0.2in}
\begin{align}
u^* \left( \nu \right)
= \mathop{\mbox{argmin}}_{u}
    \max_{i}
    \sum_{y} p_i^{u} \left( y \right)
    \left(
    \hspace{-0.04in}
    \frac{
    1 + \epsilon ( J\, p_i^{u} \left( y \right) - 1 )
    }
    {
    1 + \epsilon (
    \sum\limits_{ j \neq i }
    	 \hspace{-0.02in}
      \frac{J \nu \left( j \right)  p_j^{u} \left( y \right)}
      {\left( 1 - \nu \left( i \right) \right)}
     - 1 )
    }
    \hspace{-0.04in}	
    \right)^{\hspace{-0.08in} -\frac{\eta}{J \epsilon}}.
\label{Pf-Lemma1-9}
\end{align}
If we select the control to be $u_k = u^*
\left( \nu_{k-1} \right)$, where $u^* \left( \nu
\right)$ is as in (\ref{Pf-Lemma1-9}), then we get that
for any $i \in \mathcal{M},\ k = 2, \ldots, n$, and any
realization of $\nu_{k-1}$ (as a function of
$y^{k-1}$),
\begin{align}
\mathbb{E}_i \hspace{-0.02in} \left [
        \left( \frac{L_k \left( i \right)}
        {L_{k-1} \left( i \right)} \right)^{
        \hspace{-0.05in}-\frac{\eta}{J \epsilon}}
        \,\Bigg \vert \mathcal{F}_{k-1}
        \right ]
\hspace{1.8in} \nonumber \\
\ =\
 \min_u
 \max_{i}\,
    \sum_{y} p_i^{u} \left( y \right)
    \left (
    \hspace{-0.05in}
    \frac{
    1 + \epsilon ( J\, p_i^{u} \left( y \right) - 1 )
    }
    {
    1 + \epsilon (
    \sum\limits_{ j \neq i }
     \frac{J\, \nu_{k-1} \left( j \right)  p_j^{u} \left( y \right)}
     {\left( 1 - \nu_{k-1} \left( i \right) \right)}
     - 1 )
    }
    \hspace{-0.04in}
    \right )^{\hspace{-0.09in} -\frac{\eta}{J \epsilon}}.
    \nonumber
\end{align}

\noindent
Note that since $\nu_0$ (uniform) and all $q_i^u,\ i
\in \mathcal{M},\ u \in \mathcal{U}$, have full
supports (cf. (\ref{Pf-Lemma1-7}) upon noting that
$\epsilon < 1$), it follows that for every $k = 1,
\ldots, n,$ and every realization $y^k,\ \nu_k \left(
y^k \right)$ will have a full support.  With this
observation, continuing from (\ref{Pf-Lemma1-6}) by
using the smoothing property of conditional
expectation, we get that
\begin{align}
e^{-\betaC}
\,\leq\,  \left
( \max_{i \in \mathcal{M}}\ \mathbb{P}_i \left \{
\delta \neq i \right \} \right )^{\hspace{-0.03in} \frac{1}{n}}
\hspace{1.7in} \nonumber \\
\leq\,  \sup_{\nu}\min_u\max_{i}
    \sum_{y} p_i^{u} \left( y \right)
    \left (
    \frac{
    1 + \epsilon ( J\, p_i^{u} \left( y \right) - 1 )
    }
    {
    1 + \epsilon (
    \sum\limits_{ j \neq i }
     \frac{ J\, \nu \left( j \right)  p_j^{u} \left( y \right)}
     {\left( 1 - \nu \left( i \right) \right)}
     - 1 )
    }
    \right )^{\hspace{-0.07in} -\frac{\eta}{J \epsilon}}
    \nonumber \\
\times \left( M-1
\right)^{\frac{\eta}{n J \epsilon}}. \nonumber
\end{align}
The lemma follows by taking the limit as $n \rightarrow
\infty$.

\subsection{Proofs of Results in Section \ref{section-SS}}
\label{app-B}

\subsubsection{The Converse Proof of Theorem \ref{Thm6}}
\label{app-B-1}

We now prove the assertion
(\ref{eqn-SS-asymp-opt-stoppingtime-CS-risk}). To
simplify notation let
\begin{align}
d_i^* &\ \triangleq\ \max\limits_{q (u)} \,\min\limits_{
j \in{\cal M} \backslash \left \{ i \right \} }
\,\sum\limits_{u } ~q(u)D(p^u_i \| p^u_j),
\nonumber \\
Z_{ij}(n) &\ \triangleq\
\log\frac{p_i(Y^n,U^n)}{p_j(Y^n,U^n)}
\nonumber
\end{align}
It is not hard to see that
(\ref{eqn-SS-asymp-opt-stoppingtime-CS-risk}) follows
immediately from Lemma \ref{lem:stop-time-is-large_whp}
below and Markov inequality.
\begin{lemma}
\label{lem:stop-time-is-large_whp}For every $0<\rho<1,$
any sequence of tests with vanishing maximal 
risk
i.e., 
$\max\limits_{k \in \mathcal{M}} \,R_k ~\rightarrow~0$,
satisfies
\begin{align}
\mathbb{P}_i\Big\{N>\frac{-\rho\log
R_i}{d_i^*}\Big\}\rightarrow 1, \nonumber
\end{align}
for every $i \in \mathcal{M}$.
\end{lemma}

Lemma \ref{lem:stop-time-is-large_whp} in turn relies
on the following lemma.
\begin{lemma}
\label{lem:LLR_large_whp} For any sequence of tests with 
$\max\limits_{k \in \mathcal{M}} \,R_k ~\rightarrow~0,$
any $0<\rho<1$, it holds that for each $j \in {\cal M},$
\begin{align}
\mathbb{P}_i\Big\{Z_{ij}(N)\geq -\rho\log
R_i\Big\}\rightarrow 1.
\end{align}
\end{lemma}
\begin{proof}
Define the subset $Q_n$ of the sample space as
\[
Q_n = \{(y^n,u^n):Z_{ij}(n)<\rho\log R_i,~\delta = i,
N=n\}
\]
From the definition of $R_i$ in
(\ref{eqn-SS-def-risks}), for every $j\in{\cal
M}\setminus\{i\}$, we have the following set of
inequalities
\begin{align}
\frac{R_i}{\pi(j)} ~\geq~ \mathbb{P}_j\{\delta = i\}
~\geq~ \sum_{n=1}^\infty\mathbb{P}_j \left \{ Q_n
\right \} ~\geq~
R_i^\rho \sum_{n=1}^\infty \mathbb{P}_i\{Q_n\},
\nonumber
\end{align}
where the third inequality follows from the fact that
$Z_{ij}(n)< -\rho\log R_i$ on $Q_n$. Hence, for every
$i \neq j,\ i, j \in \mathcal{M},$
\begin{align}
\label{eq:prob_Q_n} \sum_{n=1}^\infty
\mathbb{P}_i\{Q_n\} \ \leq\
\frac{R_i^{1-\rho}}{\pi(j)}.
\end{align}
Thus,
\begin{align}
\label{eq:LLR_has_to_be_large}
\mathbb{P}_i\{Z_{ij}(N)<-\rho\log R_i\} 
&~\leq~\sum_{n=1}^\infty \mathbb{P}_i\{Q_n\} 
	~+~  \mathbb{P}_i\{\delta\ne i\} \nonumber
\end{align}
\begin{align}
&~\leq~
\frac{R_i^{1-\rho}}{\pi(j)} + \sum_{j\in{\cal M}\setminus\{i\}}\frac{R_j}{\pi(i)}.
\end{align}
The second inequality above follows from
(\ref{eq:prob_Q_n}) and from the fact that
$\mathbb{P}_{i}(\delta = j)\leq \frac{R_j}{\pi(i)}$.
The right-side of (\ref{eq:LLR_has_to_be_large}) goes
to $0$ since $R_i\rightarrow 0,$ for each $i \in {\cal
M}$. This proves Lemma \ref{lem:LLR_large_whp}.
\end{proof}

The following result follows from a standard
martingale convergence argument as in Lemma 5 in
\cite{cher-amstat-1959} and is omitted due to space
constraints. \newline \indent
{\em For any $0 < \rho < 1,$ it holds that}
\begin{align}
\label{eq:lemma5_chernoff} \lim_{n \rightarrow \infty}
\mathbb{P}_i\{\max_{1\leq m\leq n}\min_{j\in{\cal
M}\setminus\{i\}}Z_{ij}(m)\geq n(d_i^* +
1-\rho)\}\rightarrow ~=~0.
\end{align}

Combining the result in Lemma \ref{lem:LLR_large_whp}
and (\ref{eq:lemma5_chernoff}), we get for every 

\noindent
$0 < \rho < 1,$
\begin{align}
\mathbb{P}_i\Big\{N\leq\frac{-\rho\log
R_i}{d_i^*+1-\rho}\Big\} ~\rightarrow~ 0,
\end{align}
which is equivalent to the assertion of Lemma
\ref{lem:stop-time-is-large_whp}. 

\vspace{0.1in}
\subsubsection{The Achievability Proof of Theorem \ref{Thm6} without Condition (\ref{eqn-SS-positivity-cond})}
\label{app-B-2}

Because the instantaneous control picked in
(\ref{eqn-SS-opt-random-control}) is a function only of
the identity of the ML estimate of the hypothesis and
not of the reliability of the estimate, e.g., the value
of the posterior probability of the ML hypothesis, when
the ML estimate is incorrect, the instantaneous
control in (\ref{eqn-SS-opt-random-control}) can be
quite bad.  This can happen with large probability
especially when only a few observations are collected.
Condition (\ref{eqn-SS-positivity-cond}) 
essentially
ensures that
when the ML hypothesis is incorrect, the control value
of (\ref{eqn-SS-opt-random-control}) will not be too
bad.  Consequently, (\ref{eqn-SS-positivity-cond})
leads to a fast convergence of the ML estimate of the
hypothesis to the true one when the ML estimation is
used together with the control policy
(\ref{eqn-SS-opt-random-control}) at all times. Without
(\ref{eqn-SS-positivity-cond}), the convergence may not
happen or even if it does, it may not be fast enough.
This phenomenon is analogous to and is tightly
connected to another one, which occurs in a
somewhat more exacerbated form, in stochastic adaptive
control \cite{kuma-vara-eit-book-2006} illustrating the
failure of ML identification in closed-loop
\cite{bork-vara-ieeetac-1979}.
\newline \indent
As previously mentioned at the end of Section \ref{subsection-SS-test}, 
we slightly modify the control policy (\ref{eqn-SS-opt-random-control}) by 
occasionally sampling from the uniform control independently of the identity 
of the ML hypothesis; this sparse sampling is used to guard against the 
event of incorrect ML estimation of the hypothesis.
Precisely, for some $a > 1,$ at times $k = \lceil a^l \rceil,\ l = 0, 1,
\ldots,$ we let $U_{k+1}$ be uniformly distributed on
$\mathcal{U}$.  At all other times, we still follow the
control policy in (\ref{eqn-SS-opt-random-control}).
The stopping rule is still as in
(\ref{eqn-SS-stopping-rule}), and the final decision is
still ML.  Without loss of generality, we can assume
that for every $i \neq j,\ i, j \in \mathcal{M}$, there
exists a $u \in \mathcal{U}$ for which
\begin{align}
    D \left( p_i^u \| p_j^u \right) ~>~ 0,
    \label{eqn-AppB-2-1}
\end{align}
otherwise, the probability of error can never be driven
to zero. It now follows from (\ref{eqn-AppB-2-1}) and
the argument as in the proof of \cite[Lemma
1]{cher-amstat-1959} that for every $i \neq j$, and all
$n$ sufficiently large
\begin{align}
    \mathbb{P}_i \left \{
        \sum_{k=1}^n L_k
    \right \} ~\leq~ e^{-b \frac{\log{n}}{\log{a}}},
    \nonumber
\end{align}
where $L_k \triangleq
        \log{ \Big(
                    \frac{p_i^{U_k}\left( Y_k \right)}
                    {p_j^{U_k}\left( Y_k \right)}
              \Big )}$,
for some $b > 0,$ as we can only guarantee that
$\mathbb{E}_i \Big [ e^{- \frac{1}{2} L_k} \Big \vert
\mathcal{F}_{k-1} \Big ] < 1$ for
$\frac{\log{n}}{\log{a}}$ times in $n$ time slots
(precisely at those times when the control value is
forced to be uniformly distributed).  Let $T$ be the
earliest time such that the ML estimate of the
hypothesis equals the true hypothesis for all time $k
\geq T.$  Then, we get that for all sufficiently large
$k,$
\begin{align}
    \mathbb{P}_i \left \{ T > k \right \} ~\leq~
    M \sum_{t \geq k} e^{-b \frac{\log{t}}{\log{a}}}
    ~\leq~ O \left( k^{-\gamma} \right)
    \label{eqn-AppB-2-unconditional-consistency}
\end{align}
for an arbitrary large $\gamma$ when $a$ is chosen to
be sufficiently close to 1.  Note that it was shown in
\cite[Lemma 1]{cher-amstat-1959} that if
(\ref{eqn-SS-positivity-cond}) holds, then
$\mathbb{P}_i \left \{ T > k \right \}$ decays
exponentially.
\newline \indent
Our achievability proof of asymptotic optimality
without (\ref{eqn-SS-positivity-cond}), i.e., that the
modified test satisfies
(\ref{eqn-SS-asymp-opt-stoppingtime-CS}) without
imposing (\ref{eqn-SS-positivity-cond}), follow closely
the steps in the proof of \cite[Lemma 2]
{cher-amstat-1959} under assumption
(\ref{eqn-SS-positivity-cond}).  Due to space
limitations, we shall just emphasize key steps and
point out the difference from the proof when
(\ref{eqn-SS-positivity-cond}) is relaxed.  To this
end, we denote the maximizers in the denominator on the
right-side of (\ref{eqn-SS-asymp-opt-stoppingtime-CS})
by $q_i^*(u)$.
\newline \indent
Referring to the stopping rule in
(\ref{eqn-SS-stopping-rule}), we see that the stopping
time depends on the time needed for the Log-Likelihood
Ratio (LLR) corresponding to the closest alternative
hypothesis to cross the stopping threshold $-\log c$.
Thus, the main idea is to show that the LLR per
observation concentrates around the denominator on the
right-side of (\ref{eqn-SS-asymp-opt-stoppingtime-CS})
for the control policy described above. The key step in
the proof of (\ref{eqn-SS-asymp-opt-stoppingtime-CS})
deals with the following decomposition for an arbitrary
hypothesis $j \neq i,$ where $i$ is the true
hypothesis,
\begin{align}
\frac{1}{n} \sum_{k=1}^n L_k 
&= 
\frac{1}{n}
\sum_{k=1}^n 
\left \{ L_k - \mathbb{E}_i \left [ L_k \big \vert \mathcal{F}_{k-1} \right ] \right \} \nonumber \\
&  \ \ +
\frac{1}{n} \sum_{k=1}^n \left \{ \mathbb{E}_i \left [
L_k \big \vert \mathcal{F}_{k-1} \right ] - \sum_{u}
q_i^* \left( u \right) D \left( p_i^u \| p_j^u \right)
\right \} \nonumber \\
&  \ \ + \sum_{u} q_i^* \left( u \right) D \left( p_i^u
\| p_j^u \right). \label{eqn-Pf-Thm3-llr-decomp1}
\end{align}
The proof of the measure concentration then boils down
to proving that the two averages of the bracketed
$\left \{ \right \}$-quantities concentrate around 0
from the negative side with a sufficiently quick decay
of the probability of non-concentration. In particular,
it suffices to prove that the following two sequences
of probabilities (as a function of $n$) go to zero
sufficiently fast:
\begin{align}
\mathbb{P}_i \left \{ \frac{1}{n} \sum_{k=1}^n \left \{
L_k - \mathbb{E}_i \left [ L_k \big \vert
\mathcal{F}_{k-1} \right ] \right \} ~<~ -\epsilon
\right \}, \label{eqn-Pf-Thm3-conc-llr-1}
\end{align}
and
\begin{align}
\mathbb{P}_i \left \{ \frac{1}{n} \sum_{k=1}^n \left \{
\mathbb{E}_i \left [ L_k \big \vert \mathcal{F}_{k-1}
\right ] - \sum\limits_{u} q_i^* \left( u \right) D
\left( p_i^u \| p_j^u \right) \right \} \,<\, - \epsilon
\right \}. \label{eqn-Pf-Thm3-conc-llr-2}
\end{align}
Note that the minimum value of the third term in the
decomposition in (\ref{eqn-Pf-Thm3-llr-decomp1}) over
$j \neq i$ is specifically the denominator on the
right-side of (\ref{eqn-SS-asymp-opt-stoppingtime-CS}).
\newline \indent
The same argument leading to \cite[Equation
(5.10)]{cher-amstat-1959} gives that
(\ref{eqn-Pf-Thm3-conc-llr-1}) goes to zero
exponentially. Also,
(\ref{eqn-AppB-2-unconditional-consistency}) implies a
polynomial decay of (\ref{eqn-Pf-Thm3-conc-llr-2}), as
with probability 1,
\begin{align}
\Bigg \vert \sum_{k=1}^n \left \{
\mathbb{E}_i \left [ L_k \big \vert \mathcal{F}_{k-1}
\right ]
~-~ \sum\limits_{u} q_i^* \left( u \right) D \left(
p_i^u \| p_j^u \right)
\right \} \Bigg \vert \hspace{0.5in} \nonumber \\
\leq\   C' \min(T, n) + C''
\log{n}, \nonumber
\end{align}
for some constants $C', C''$ by virtue of fact that
$q=q_i^*$ for each $k \geq T,$ such that $k \neq \lceil
a^l \rceil, l \geq 1$ (cf. the definition of $T$ in
above).  This will lead us to \cite[Equation
(5.9)]{cher-amstat-1959} but only with a polynomial
decay (with an arbitrarily high degree $\gamma$ in
(\ref{eqn-AppB-2-unconditional-consistency})) in the
probability on the right-side of the equation.
Nevertheless, the sufficiently quick polynomial decay
in the probability therein still enables us to complete
the steps at the beginning to proof of \cite[Lemma
2]{cher-amstat-1959} to eventually upper bound the
asymptotes of the expected sample sizes to be
(\ref{eqn-SS-asymp-opt-stoppingtime-CS}).

\vspace{0.1in}
\subsubsection{Proof of Theorem \ref{Thm7}}
\label{app-B-3}

We first prove
(\ref{eqn-Thm7-hard-risk-constraints}). Let $i$ be the
true hypothesis. For any $j\in{\cal M}, j\ne i$,
consider the event
\begin{align}
A_{n,j} = \{(y^n,u^n): N_A=n,\ \delta=j\}.
\end{align}
Following the approach in \cite{atia-veer-2012}, on the
set $A_{n,j}$ we have the following set of
inequalities,
\begin{align}
\log\left(\frac{\pi(j)p_j(y^n,u^n)}{\pi(i)p_i(y^n,u^n)}\right)
&\geq\ 
\log\left(\frac{\pi(j)p_j(y^n,u^n)}
{\max\limits_{i \ne j}\pi(i)p_i(y^n,u^n)}\right) 
\nonumber \\  
&\geq\ 
\log\left(\frac{(M-1) \pi(j)}{\bar{R}_j}\right).
\end{align}
The last inequality above 
follows since the test ends at $n$ and
the stopping criteria must be met for the choice of the
thresholds in
(\ref{eqn-SS-stopping-rule-hard-risk-constraints}).
Thus,
\begin{align}
\mathbb{P}_i\{A_{n,j}\} \ &\leq\
\frac{\bar{R}_j}{(M-1)\pi(i)} \mathbb{P}_j\{A_{n,j}\}.
\nonumber
\end{align}
It now follows that
\begin{align}
\mathbb{P}_i \left \{\delta = j \right \} \,=\,
\sum_{n=1}^\infty \mathbb{P}_i\{A_{n,j}\} &\,\leq\,
\frac{\bar{R}_j}{(M-1)\pi(i)}\sum_{n=1}^\infty
\mathbb{P}_j\{A_{n,j}\} \nonumber \\
&\,\leq\,
\frac{\bar{R}_j}{(M-1)\pi(i)}.
\label{eqn-Pf-Appendix-B3-1}
\end{align}
From the definition of $R_j$ in
(\ref{eqn-SS-def-risks}), we then get that $R_j\leq
\bar{R}_j$. The result holds for each $j \in
\mathcal{M}$

The last assertion of Theorem \ref{Thm7}
pertaining to asymptotic optimality of the proposed
test follows by considering yet another test with the
stopping rule
(\ref{eqn-SS-stopping-rule-hard-risk-constraints})
being replaced by the following stopping rule with a
single threshold
\begin{equation}
    \log{ \left(
    \frac{
        p_{\hat{i}_n} \left( y^n, u^n \right)
    }
    {
        \max\limits_{j \,\neq\, \hat{i}_n} ~p_{j} \left( y^n, u^n \right)
    } \right)
    }
    ~\geq~
    \log{\left( \left( M-1 \right)
              \left( \max\limits_{i \neq j}
                \frac{\pi (j) }{\bar{R}_{i}}
          \right) \right)},
\label{eqn-Pf-Appendix-B3-2}
\end{equation}
and with the same control and decision rule as those of
the proposed test. It follows from
(\ref{eqn-SS-stopping-rule-hard-risk-constraints}) and
(\ref{eqn-Pf-Appendix-B3-2}) that the stopping time of
this new test will always dominate (larger than) that
of the proposed test a.s.  
Let us denote the two respective stopping times by $N$ and $N'$.
Since $\pi$ has a full
support, as $\max\limits_{i \in \mathcal{M}} \bar{R}_i
\rightarrow 0,$ the single threshold on the right-side
of (\ref{eqn-Pf-Appendix-B3-2}) will go to infinity. By
Theorem \ref{Thm6}, this new test with the single
threshold is asymptotically optimal, i.e., it satisfies,
for every $i \in \mathcal{M},$
\begin{align}
\lim\limits_{\max\limits_i \bar{R}_i \rightarrow 0}
\ -\frac{\mathbb{E}_i \left [ N' \right ]}{\log{c}}
\ \leq\ 
\frac{1}{\max\limits_{q(u)} \min\limits_{j \neq i} \sum\limits_u q(u) D \left( p_i^u \| p_j^u \right)},
\label{eqn-Pf-Appendix-B3-3}
\end{align}
where $c = \frac{1}{M-1} \left( \min\limits_{i \neq j}  \frac{\bar{R}_i}{\pi(j)} \right).$
On the other hand, it follows from (\ref{eqn-Pf-Appendix-B3-1}) and 
the assumption in the statement of Theorem \ref{Thm7} that 
\begin{align}
	\max\limits_i \mathbb{P}_i \left \{ \delta \neq i \right \}
	~\leq~ \max\limits_{i \neq j} \frac{\bar{R}_j}{\pi(i)}
	~\leq~ K'c,
\label{eqn-Pf-Appendix-B3-4}
\end{align}
for a suitable constant $K'$.  The aforementioned dominance, i.e., $N \leq N'$ a.s., and 
(\ref{eqn-Pf-Appendix-B3-4}) along with 
(\ref{eqn-Pf-Appendix-B3-3}) give that for every $i \in \mathcal{M},$
\begin{align}
\lim_{\max\limits_{i} \bar{R}_i \rightarrow 0}\,
-\frac{\mathbb{E}_i \left [ N \right ]}
{\log{\left(\max\limits_k \mathbb{P}_k \left \{ \delta \neq k \right \}\right)}}
&\leq\  
\lim_{\max\limits_{i} \bar{R}_i \rightarrow 0}\,
-\frac{\mathbb{E}_i \left [ N' \right ]}{\log{c}}	\hspace{0.6in}  \nonumber
\end{align}
\begin{align}
\hspace{1.6in} \leq\   
\frac{1}{\max\limits_{q(u)} \min\limits_{j \neq i} ~\sum\limits_{u} q(u) D\left( p_i^u \| p_j^u \right)}.
\nonumber
\end{align}

\bibliographystyle{IEEEtran}
\bibliography{niti-atia-veer-2012}

\vspace{0.1in}
Sirin Nitinawarat (SM'09--M'11) obtained the B.S.E.E. degree from Chulalongkorn University, Bangkok, Thailand, with first class honors, and the M.S.E.E. degree from the University of Wisconsin, Madison.  He received his Ph.D. degree from the Department of Electrical and Computer Engineering and the Institute for Systems Research at the University of Maryland, College Park, in December 2010.  He is now a postdoctoral research associate at the Coordinated Science Laboratory at the University of Illinois at Urbana-Champaign.  His research interests are in stochastic control, statistical signal processing, information and coding theory, communications, estimation and detection, and machine learning.  

Dr. Nitinawarat was a co-organizer for the special session on {\em Controlled Sensing for Inference} at the 2012 IEEE International Conference on Acoustics, Speech and Signal Processing (ICASSP); and chair for the session on {\em Distributed Inference in Sensor Networks} at the $49^th$ Annual Allerton Conference on Communication, Control, and Computing (2011).  He was a finalist for the best student-paper award at the IEEE International Symposium on Information Theory, which was held at Austin, Texas, in 2010.  During his Ph.D. study at the University of Maryland, he received graduate teaching fellowships in Fall 2007, Spring 2008, and Spring 2009.

\vspace{0.1in}
George K. Atia (S'01--M'09) received the B.Sc. and M.Sc. degrees from Alexandria University, Egypt, in 2000 and 2003, respectively, and the Ph.D. degree from Boston University, MA, in 2009, all in electrical and computer engineering.

He joined the University of Central Florida in Fall 2012 where he is currently an assistant professor in the department of Electrical Engineering and Computer Science. From Fall 2009 to 2012, he was a postdoctoral research associate at the Coordinated Science Laboratory (CSL) at the University of Illinois at Urbana-Champaign (UIUC). His research interests include statistical signal processing, stochastic control, wireless communications, detection and estimation theory, and information theory.

Dr. Atia is the recipient of many awards, including the Outstanding Graduate Teaching Fellow of the Year Award in 2003--2004 from the Electrical and Computer Engineering Department at Boston University, the 2006 College of Engineering DeanÕs Award at the BU Science and Engineering Research Symposium, and the best paper award at the International Conference on Distributed Computing in Sensor Systems (DCOSS) in 2008.

\vspace{0.1in}
Venugopal V. Veeravalli (M'92--SM'98--F'06) received the B.Tech. degree (Silver Medal Honors) from the Indian Institute of Technology, Bombay, in 1985, the M.S. degree from Carnegie Mellon University, Pittsburgh, PA, in 1987, and the Ph.D. degree from the University of Illinois at Urbana-Champaign, in 1992, all in electrical engineering.

He joined the University of Illinois at Urbana-Champaign in 2000, where he is currently a Professor in the department of Electrical and Computer Engineering and the Coordinated Science Laboratory. He served as a Program Director for communications research at the U.S. National Science Foundation in Arlington, VA from 2003 to 2005.  He has previously held academic positions at Harvard University, Rice University, and Cornell University, and has been on sabbatical at MIT, IISc Bangalore, and Qualcomm, Inc. His research interests include distributed sensor systems and networks, wireless communications, detection and estimation theory, including quickest change detection, and information theory.

Prof. Veeravalli was a Distinguished Lecturer for the IEEE Signal Processing Society during 2010--2011. He has been an Associate Editor for Detection and Estimation for the IEEE Transactions on Information Theory and for the IEEE Transactions on Wireless Communications. Among the awards he has received for research and teaching are the IEEE Browder J. Thompson Best Paper Award, the National Science Foundation CAREER Award, and the Presidential Early Career Award for Scientists and Engineers (PECASE).

\end{document}